\documentclass[11pt,pdfa]{article}
\usepackage[in]{fullpage}

\usepackage{iftex}
\ifPDFTeX
  \usepackage[utf8]{inputenc}
  \usepackage[noTeX]{mmap}
  \usepackage[T1]{fontenc}
\fi
\ifLuaTeX
  \usepackage{luatex85}
  \usepackage[noTeX]{mmap}
\fi

\usepackage[style=alphabetic,minalphanames=3,maxalphanames=4,maxnames=99,backref=true]{biblatex}

  \def\bool{\{0,1\}}
\def\A{{\algo A}}
\def\B{\mathcal{B}}
\def\D{\mathcal{D}}

\def\xor{\oplus}

\usepackage{mdframed}
\usepackage{dsfont}
\usepackage{tabularx}
\usepackage[normalem]{ulem}
\usepackage{comment}
\usepackage[shortlabels]{enumitem}
\usepackage{bm}
\usepackage{amsfonts}
\usepackage{xspace}
\usepackage{amsmath}
\usepackage{amssymb}
\usepackage{amsthm}
\usepackage{xcolor}
\usepackage{graphicx}
\usepackage{qtree}
\usepackage{tree-dvips}
\usepackage{float}
\usepackage{hyperref}
\usepackage{physics}
\usepackage{authblk}
\usepackage{braket}
\usepackage{mathtools}
\usepackage{mathrsfs}
\usepackage{tikz}
\usepackage[linesnumbered,ruled,vlined]{algorithm2e}
\usepackage{qcircuit}
\usepackage[nameinlink,capitalize]{cleveref}

\newcommand{\subversion}[1]{}

  \hypersetup{colorlinks={true},linkcolor={blue},citecolor=magenta}

  \theoremstyle{plain}
  \newtheorem{theorem}{Theorem}[section]
  \newtheorem{lemma}[theorem]{Lemma}
  \newtheorem{corollary}[theorem]{Corollary}
  \newtheorem{definition}[theorem]{Definition}

\usepackage[style=alphabetic,minalphanames=3,maxalphanames=4,maxnames=99,backref=true]{biblatex}

\newtheorem*{rep@theorem}{\rep@title}
\newcommand{\newreptheorem}[2]{%
\newenvironment{rep#1}[1]{%
 \def\rep@title{#2 \ref{##1}}%
 \begin{rep@theorem}}%
 {\end{rep@theorem}}}
\makeatother

\newreptheorem{theorem}{Theorem}
\newreptheorem{lemma}{Lemma}
\newreptheorem{corollary}{Corollary}

\newmdtheoremenv[backgroundcolor=gray!10,
                 linewidth=0pt,
                 innerleftmargin=4pt,
                 innerrightmargin=4pt,
                 innertopmargin=2pt,
                 innerbottommargin=4pt,
            splitbottomskip=4pt]{protocol}[prot]{Protocol}

\newcommand{\ignore}[1]{}

\newcommand{\N}{\mathbb{N}}

\newcommand{\negl}{\mathsf{negl}}

\newcommand\id{\mathbb{I}}

\definecolor{webgreen}{rgb}{0,.5,0}
\definecolor{webblue}{rgb}{0,0,.5}
\newcommand\algo{\mathcal}

\newcommand{\Enc}{\ensuremath{\mathsf{Enc}}\xspace}
\newcommand{\Dec}{\ensuremath{\mathsf{Dec}}\xspace}

\newcommand{\Gen}{\ensuremath{\mathsf{Gen}}\xspace}

\newcommand{\expref}[2]{\texorpdfstring{\hyperref[#2]{#1~\ref{#2}}}{#1~\ref{#2}}}
\newcommand{\bits}{\{0,1\}}
\newcommand{\bit}{\{0,1\}}

\newcommand{\from}{\leftarrow}

\newcommand{\QRACVL}{\ensuremath{\mathsf{QRAC\mbox{-}VL}}\xspace}

\newcommand{\swap}{{\sf swap}}

\newcommand{\Perms}{\mathcal{S}}

\newif\ifsubmission
\submissionfalse
\ifsubmission
\newcommand{\ga}[1]{}
\newcommand{\ap}[1]{}
\newcommand{\ks}[1]{}
\newcommand{\chen}[1]{}
\else
\definecolor{mgreen}{rgb}{.1,.7,0}
  
\newcommand{\ga}[1]{{\noindent \textcolor{purple}{\emph{(GA:  #1)}}}{}}
\newcommand{\ap}[1]{{\noindent \textcolor{blue}{\emph{(AP:  #1)}}}{}}
\newcommand{\ks}[1]{{\noindent \textcolor{mgreen}{\emph{(KS:  #1)}}}{}}
\newcommand{\chen}[1]{{\noindent \textcolor{orange}{\emph{(Chen:  #1)}}}{}}
\fi

\newcommand{\spi}{\textsf{SPI}\xspace}
\newcommand{\aspi}{\textsf{aSPI}\xspace}
\newcommand{\dpi}{\textsf{DPI}\xspace}
\newcommand{\adpi}{\textsf{aDPI}\xspace}

\newcommand{\poly}{\mathrm{poly}}
\newcommand{\Si}{\mathsf{S}\xspace}

\newcommand{\Di}{\mathsf{D}\xspace}

\newcommand{\aS}{\mathsf{aS}\xspace}
\newcommand{\aD}{\mathsf{aD}\xspace}
\newcommand{\Opi}{\mathcal{O}_\pi}
\newcommand{\Opii}{\mathcal{O}_{\pi^{-1}_{\bot y}}}
\newcommand{\Opiii}{\pi_{\bot y}}

\newcommand{\DecisionInvert}{\ensuremath{\mathsf{DecisionInvert}}\xspace}

\title{On the Two-sided Permutation Inversion Problem}

\author[1]{Gorjan Alagic}
\author[2]{Chen Bai}
\author[3]{Alexander Poremba}
\author[4]{Kaiyan Shi}

\affil[1]{QuICS, University of Maryland, and NIST}
\affil[2]{Dept.\ of Electrical and Computer Engineering, University of Maryland}
\affil[3]{Computing and Mathematical Sciences, California Institute of Technology}
\affil[4]{Dept.\ of Computer Science, University of Maryland}

\begin{document}
\date{\vspace{-5ex}}

\maketitle

\begin{abstract}
In the permutation inversion problem, the task is to find the preimage of some challenge value, given oracle access to the permutation. This is a fundamental problem in query complexity, and appears in many contexts, particularly cryptography. In this work, we examine the setting in which the oracle allows for quantum queries to both the forward and the inverse direction of the permutation---except that the challenge value cannot be submitted to the latter. Within that setting, we consider two options for the inversion algorithm: whether it can get quantum advice about the permutation, and whether it must produce the entire preimage (search) or only the first bit (decision). We prove several theorems connecting the hardness of the resulting variations of the inversion problem, and establish a number of lower bounds. Our results indicate that, perhaps surprisingly, the inversion problem does not become significantly easier when the adversary is granted oracle access to the inverse, provided it cannot query the challenge itself.
\end{abstract}

\section{Introduction} 

\subsection{The permutation inversion problem} 

The permutation inversion problem is defined as follows: given a permutation $\pi: [N]\rightarrow [N]$ and an image $y \in [N]$, output the correct preimage $x := \pi^{-1}(y)$. In the decision version of the problem, it is sufficient to output only the first bit of $x$. If the algorithm can only access $\pi$ by making classical queries, then making $T = \Omega(N)$ queries is necessary and sufficient for both problems. If quantum queries are allowed, then Grover's algorithm can be used to solve both problems with $T = O(\sqrt{N})$ queries~\cite{grover1996fast,ambainis2002quantum}, which is worst-case asymptotically optimal~\cite{bennett1997strengths,ambainis2002quantum,nayak2010inverting}. 

In this work, we consider the permutation inversion problem in a setting where the algorithm is granted both forward and inverse quantum query access to the permutation $\pi$. In order to make the problem nontrivial, we modify the inverse oracle so that it outputs a reject symbol when queried on the challenge image $y$. We call this the \emph{two-sided permutation inversion problem}. This variant appears naturally in the context of chosen-ciphertext security for encryption schemes based on (pseudorandom) permutations~\cite{katz2020introduction}, as well as in the context of sponge hashing (SHA3)~\cite{guido2011cryptographic}.
We consider several variants:
\begin{enumerate}
\item \emph{(Auxiliary information.)} With this option enabled, the inversion algorithm now consists of two phases. The first phase is given a full description of $\pi$ (e.g., as a table) and allowed to prepare an arbitrary quantum state $\rho_\pi$ consisting of $S$ qubits. This state is called \emph{auxiliary information} or \emph{advice}. The second phase of the inversion algorithm is granted only the state $\rho_\pi$ and query access to $\pi$, and asked to invert an image $y$. The two phases of the algorithm can also share an arbitrarily long uniformly random string, referred to as \emph{shared randomness}. The complexity of the algorithm is measured in terms of the number of qubits $S$ of the advice state (generated by the first phase) and the total number of queries $T$ (made during the second phase.) 
\item \emph{(Adaptive restriction of challenge distribution.)} In this case, the inversion algorithm again consists of two phases. The first phase is again given a full description of $\pi$, and allowed to output a string $\mu \in \{0,1\}^m$ for $m < n$, where $n = \sqrt{N}$. The second phase is then granted query access to $\pi$ and asked to invert an image $y$ which is sampled uniformly at random from the set of all strings whose last $m$ bits equal $\mu$.
\item \emph{(Search vs Decision.)} Here the two options simply determine whether the inversion algorithm is tasked with producing the entire preimage $x = \pi^{-1}(y)$ of the challenge $y$ (search version), or only the first bit $x_0$ (decision version.)
\end{enumerate}

If the algorithm is solving the search problem, we refer to it as a search permutation inverter, or \spi. If it is solving the decision problem, we refer to it as a decision permutation inverter, or \dpi. If an \spi uses $S$ qubits of advice and $T$ queries to succeed with probability at least $\epsilon$ in the search inversion experiment, we say it is a $(S, T, \epsilon)$-\spi. If a \dpi uses $S$ qubits of advice and $T$ queries to succeed with probability at least $1/2 +\delta$ in the decision inversion experiment, we say it is a $(S, T, \delta)$-\dpi. If the algorithm is allowed to adaptively restrict the challenge distribution, we say it is adaptive and denote it by $\aspi$ or $\adpi$, as appropriate.

In this work, we are mainly interested in the \emph{average-case} setting. This means that both the permutation $\pi$ and the challenge image $y$ are selected uniformly at random. Moreover, the success probability is taken over all the randomness in the inversion experiment, i.e., over the selection of $\pi$ and $y$ along with all internal randomness and measurements of the inversion algorithm.

In \Cref{sec:preliminaries}, we present technical preliminaries, including the swapping lemma and quantum random access codes (QRAC), for subsequent proof. In \Cref{sec:perminv}, we introduce several definitions of the permutation inversion problem,  with both auxiliary information and adaptive restriction of challenge distribution. Within \Cref{sec:amplification}, we show methods for amplifying the success probability of inversion in the non-adaptive case. Subsequently, in \Cref{sec:reductions}, we illustrate two reductions: from search-to-decision with auxiliary information and from unstructured search-to-decision without auxiliary information. These reductions are then utilized to derive lower bounds, as shown in \Cref{sec:lower bound}. Finally, in \Cref{sec:applications}, we propose a novel security notion, called one-way-QCCRA2, and establish the security of two common schemes under this notion, subject to specific conditions.

\subsection{Related work}

Previous works have considered the quantum-query \emph{function} inversion problem~\cite{hhan2019quantum,chung2019lower,chung2020tight,dunkelman2023quantum,liu2023non}. A number of papers gave lower bounds and time-space tradeoffs for the (one-sided) quantum-query permutation inversion problem, with and without advice~\cite{ambainis2002quantum,nayak2010inverting,rosmanis2021tight,nayebi2014quantum,hhan2019quantum,chung2019lower,fefferman2015quantum,belovs2023one}. The relevant highlights among these are summarized in \Cref{tab:permutation inversion work}. 

We remark that some of these previous works~\cite{cao2021being,chung2019lower,nayebi2014quantum} do not fully address the average-case setting. Specifically, they deal with inverters that are ``restricted'' in the following manner. First, the inverter is said to ``invert $y$ for $\pi$'' if it succeeds in the inversion experiment for the specific pair $(\pi, y)$ with probability at least $2/3$. Second, the inverter is said to ``invert a $\delta$-fraction of inputs'' if $\Pr_{\pi, y}[\text{the inverter inverts }y\text{ for }\pi] \geq \delta$. This type of inverter is clearly captured by our notion above: it is an $(S, T, 2\delta/3)$-\spi. However, there are successful inverters of interest that are captured by our definition but are not restricted. For example, in a cryptographic context, one would definitely be concerned about adversaries that can invert every $(\pi, y)$ with a probability of exactly $1/n$. Such an adversary is clearly a $(S, T, 1/n)$-\spi, but is not a restricted inverter for any value of $\delta$. Other works also consider the general average-case (e.g.,~\cite{chung2020tight,liu2023non,hhan2019quantum}) but without two-way oracle access. Note that the lower bound for restricted adversaries described in \cite{nayebi2014quantum,chung2019lower} can be translated to the more general lower bound in a black box way by applying our amplification procedure described in \Cref{lem:amplification}.

\begin{table}[h]
  \small
  \begin{tabularx}{\textwidth} { 
    | >{\raggedright\arraybackslash}X 
    | >{\centering\arraybackslash}X 
    | >{\centering\arraybackslash}X 
    | >{\centering\arraybackslash}X 
    | >{\centering\arraybackslash}X 
    |}
   \hline
   & \cite{nayebi2014quantum} & \cite{chung2019lower} & \cite{hhan2019quantum} & Ours\\
   \hline
    Advice & classical & quantum & quantum & quantum\\
    \hline
    Access Type & one-sided & one-sided & one-sided & two-sided\\
    \hline
    Inverter & restricted & restricted & general & general\\
    \hline
      $T$-$S$ trade-off & $ST^2=\widetilde{\Omega}(N)$
      & $ST^2=\widetilde{\Omega}(\epsilon N)$ &  $ST^2=\widetilde{\Omega}(\epsilon^3 N)$  &  $ST^2=\widetilde{\Omega}(\epsilon^3 N)$\rule[0ex]{0pt}{3ex}\\
  \hline
\end{tabularx} 
\caption{Summary of previous work on permutation inversion with advice. Success probability is denoted by $\epsilon$. Note that $\epsilon=O(1)$ in~\cite{nayebi2014quantum}.}\label{tab:permutation inversion work}
\end{table}

To our knowledge, the two-way variant of the inversion problem has only been considered in one other work. Specifically, \cite{cao2021being} gives a lower bound of $T = \Omega(N^{1/5})$ to invert a random injective function (with two-way access and no advice) with a non-negligible success probability. 

Another novelty of our work is that we give lower bounds and time-space tradeoffs for the decision problem (rather than just search). While prior work~\cite{chung2020tight} also considered the general decision game, their generic framework crucially relies on compressed oracles \cite{zhandry2019record} which are only known to support random \emph{functions}. Consequently, their techniques cannot readily be applied in the context of permutation inversion due to a lack of ``compressed permutation oracles''.

We remark that the notion of two-way quantum accessibility to a random permutation has been considered in other works; for example,~\cite{alagic2022post,alagic2022posttwk} studied the hardness of detecting certain modifications to the permutation in this model. By contrast, we are concerned with the problem of finding the inverse of a random image.

\section{Technical preliminaries} \label{sec:preliminaries}
\subsection{Swapping Lemma}

Let $\algo A^f$ be a quantum algorithm with quantum oracle access to a function $f: \algo X \rightarrow \algo Y$, for some finite sets $\algo X$ and $\algo Y$. Let $\algo S \subseteq \algo X$ be a subset. Then, the total query magnitude of $\algo A^f$ on the set $\algo S$ is defined as
$
q(\algo A^f,\algo S)  =  \sum_{t = 0}^{T-1} \|\Pi_{\algo S}\ket{\psi_t}\|^2,
$
where $\ket{\psi_t}$ represents the state of $\algo A$ just before the $(t+1)^{\textrm{st}}$ query and $\Pi_{\algo S}$ is the projector onto $\algo S$ acting on the query register of $\algo A$. We use the following simple fact: for any subset $\algo S \subseteq \algo X$ and $\algo A$ making at most $T$ queries, it holds that $q(\algo A^f,\algo S) \leq T$. The following lemma controls the ability of a query algorithm to distinguish two oracles, in terms of the total query magnitude to locations at which the oracles take differing values.

\begin{lemma}[Swapping Lemma, \cite{Vaz98}]\label{lem:swapping} 
Let $f,g: \algo X \rightarrow \algo Y$ be functions with $f(x) = g(x)$ for all $x \notin \algo S$, where $\algo S \subseteq \algo X$.
Let $\ket{\Psi_f}$ and $\ket{\Psi_g}$ denote the final states of a quantum algorithm $\algo A$ with quantum oracle access to the functions $f$ and $g$, respectively. Then, 
$$
\| \ket{\Psi_f} - \ket{\Psi_g} \| \leq \sqrt{T \cdot q(\algo A^f,\algo S)},
$$
where $\| \ket{\Psi_f} - \ket{\Psi_g} \|$ denotes the Euclidean distance and where $T$ is an upper bound on the number of quantum oracle queries made by $\algo A$.
\end{lemma}

\subsection{Lower bounds for quantum random access codes} 

Quantum random access codes~\cite{Wiesner83, ambainis1999dense, ambainis2008quantum} are a means of encoding classical bits into (potentially fewer) qubits. We use the following variant from~\cite{chung2019lower}.  

\begin{definition}[Quantum random access codes with variable length] 
Let $N$ be an integer and let $\algo F_N = \{f: [N] \rightarrow \algo X_N \}$ be an ensemble of functions over some finite set $\algo X_N$. A quantum random access code with variable length $(\QRACVL)$ for $\algo F_N$ is a pair $(\Enc,\Dec)$ consisting of a quantum encoding algorithm $\Enc$ and a quantum decoding algorithm $\Dec$:
\begin{itemize}
    \item $\Enc(f;R)$: The encoding algorithm takes as input a function $f \in \algo F_N$ together with a set of random coins $R \in \bit^*$, and outputs a quantum state $\rho$ on $\ell = \ell(f)$ many qubits (where $\ell$ may depend on $f$).
    \item $\Dec(\rho,x;R)$: The decoding algorithm takes as input a state $\rho$, an element $x \in [N]$ and random coins $R \in \bit^*$ (same randomness used for the encoding), and seeks to output $f(x)$.
\end{itemize}
The performance of a $\QRACVL$ is characterized by parameters $L$ and $\delta$. Let $
L := \underset{f}{\mathbb{E}}[ \ell(f)]$ be the average length of the encoding over the uniform distribution on $f \in \algo F_N$, and let
$$
\delta = \underset{f,x,R}{\Pr} \left[\Dec(\Enc(f;R),x;R) = f(x)\right]
$$
be the probability that the scheme correctly reconstructs the image of the function, where $f \in \algo F_N$, $x \in [N]$ and $R$ are all chosen uniformly at random.
\end{definition}

We use the following information-theoretic lower bound on the expected length of any \QRACVL scheme for permutations, which is a consequence of \cite[Theorem 5]{chung2019lower}.

\begin{theorem}[\cite{chung2019lower}, Corollary 1]\label{cor:lower-bound-QRAC-VL}
For any \QRACVL for the set of permutations $\algo S_N$ (of the set $[N]$) with $\delta = 1-k/N$ for some $k=\Omega(1/N)$, we have 
$$
L \geq \log N! - O(k \log N)\,.
$$
\end{theorem}

\section{The permutation inversion problem} \label{sec:perminv}
We begin by formalizing the search version of the permutation inversion problem. 
We let $[N]=\{1,...,N\}$; typically we choose $N = 2^n$ for some positive integer $n$. For $f: \algo X \rightarrow \algo Y$ a function from a set $\algo X$ to an additive group $\algo Y$ (typically just bitstrings), the quantum oracle $\algo O_f$ is the unitary operator
$ \algo O_f: \ket{x}\ket{y} \rightarrow \ket{x}\ket{y \oplus f(x)}$.  
We use $\algo A^{\algo O_f}$ (or sometimes simply $\algo A^f$) to denote that algorithm $\algo A$ has quantum oracle access to $f$.

\begin{definition}\label{def:aSPI}
Let $m,n \in \N$ and $M=2^m,\ N=2^n$. An adaptive search-version permutation inverter (\aspi) is a pair $\aS = (\aS_0,\aS_1)$ of quantum algorithms, where
\begin{itemize}
    \item $\aS_0$ is an algorithm that receives as input a truth table for a permutation over $[N]$ and a random string $r$, and outputs a quantum state as well as a classical string $\mu \in \{0,1\}^m$ with $0 \leq m < n$; 
    \item $\aS_1$ is an oracle algorithm that receives a quantum state, a classical string $\mu \in \{0,1\}^m$, an image $y \in [N]$, and a random string $r$, and outputs $x \in \{0,1\}^{n-m}$.
\end{itemize}
Note that $m$ is a parameter of the adaptivity, i.e. the length of the adaptive string.
\end{definition}
\noindent We will consider the execution of an \aspi $\aS$ in the following experiment,

\begin{enumerate}
\item \emph{(sample coins)} a uniformly random permutation $\pi : [N] \rightarrow [N]$ and a uniformly random string $r \leftarrow \bit^*$ are sampled;
\item \emph{(prepare advice)}  $\aS_0$ is run, producing a pair consisting of a quantum state and a string $(\rho_{\pi,r,\mu},\mu) \leftarrow \aS_0(\pi, r)$;
\item \emph{(sample instance)} a random image $y \in [N]$ is generated by first sampling a random string $x \from \{0,1\}^{n-m}$ and then letting $y = \pi(x\| \mu)$; 

\item \emph{(invert)} $\aS_1$ is run with the oracles below, and produces a candidate preimage $x^*$.
\begin{equation}\label{eq:oracles}
    \Opi: \ket{w}\ket{z} \rightarrow \ket{w}\ket{z \xor \pi(w)}
    \qquad
    \Opii: \ket{w}\ket{z} \rightarrow \ket{w}\ket{z \xor \pi^{-1}_{\bot y}(w)},
\end{equation}
where $\pi^{-1}_{\bot y} : [N] \times \bit \rightarrow [N] \times \bit$ is defined by
\begin{align*}
    \pi^{-1}_{\bot y} (w \| b) &= \begin{cases}
        \pi^{-1}(w) \| 0 & \text{ if } b = 0 \text{ and } w \neq y\\
        1^{\lceil\log N \rceil} \| 1 & \text{ otherwise. } 
    \end{cases}
\end{align*}
To keep the notation simple, we write this entire process as $x^* \leftarrow \aS_1^{\Opiii}(\rho_{\pi, r, \mu},\mu, y, r)$. We will use $\Opiii$ to denote simultaneous access to the two oracles in \eqref{eq:oracles} throughout the paper. 
\item \emph{(check)} If $\pi(x^*\|\mu) = y$, output $1$; otherwise output $0$.
\end{enumerate}

Note that the two oracles allow for the evaluation of the permutation $\pi$ in both the forward and inverse directions. To disallow trivial solutions, the oracle outputs a fixed ``reject'' element $1^{\lceil\log N \rceil} \|1 \in [N] \times \bit$ if queried on $y$ in the inverse direction. 
\begin{definition} \label{def:eps-aSPI}
An $(S,T,\epsilon)$-\aspi is a search-version adaptive permutation inverter $\aS=(\aS_0, \aS_1)$ satisfying all of the following:
\begin{enumerate}
\item 
$
\Pr\left[\pi^{-1}(y) \leftarrow \aS_1^{\Opiii}(\rho,\mu, y, r) \,\, : \,\, (\rho,\mu) \leftarrow \aS_0(\pi, r),\ y=\pi(x\|\mu) \right] \geq \epsilon,
$ 
where the probability is taken over $\pi \leftarrow \Perms_N$, $r \leftarrow \bit^*$ and $x \gets \bits^{n-m}$, along with all internal randomness and measurements of $\aS$;
\item $S=S(\aS)$ is an upper bound on the number of qubits of $\rho$ in the above.
\item $T=T(\aS)$ is an upper bound on the number of oracle queries made by $\aS_1$.
\end{enumerate}
\end{definition}
We emphasize that the running time of $\aS$ and the length of the shared randomness $r$ are only required to be finite. We will assume that both $S$ and $T$ depend only on the parameter $N$; in particular, they will not vary with $\pi$, $y$, $r$, or any measurements. 

\begin{definition}\label{def:SPI}
A search-version permutation inverter ($\spi$) $\Si=(\Si_0,\Si_1)$ is defined as an $\aspi$ with $m=0$. An $(S,T,\epsilon)$-$\spi$ is an $(S,T,\epsilon)$-$\aspi$ with $m=0$.
\end{definition}

\paragraph{Decision version.}
The decision version of the permutation inversion problem is defined similarly to the search version above. An adaptive decision-version permutation inverter (\adpi) is denoted $\aD = (\aD_0,\aD_1)$, and outputs one bit $b$ rather than a full candidate preimage. In the ``check'' phase of the experiment, the single-bit output $b$ of $\aD_1$ is compared to the first bit $\pi^{-1}(y)|_0$ of the preimage of the challenge $y$. Success probability is now measured in terms of the advantage over the random guessing probability of $1/2$.

\begin{definition} \label{def:delta-aDPI}
  A $(S,T,\delta)$-\adpi is a decision-version adaptive permutation inverter $\aD =(\aD_0, \aD_1)$ satisfying all of the following:
\begin{enumerate}
\item 
$
  \Pr\left[\pi^{-1}(y)|_0 \leftarrow \aD_1^{\Opiii}(\rho,\mu, y, r) \,\, : \,\, (\rho, \mu) \leftarrow \aD_0(\pi, r), \ y=\pi(x\|\mu) \right] \geq \frac{1}{2} + \delta,
$ 
where the probability is taken over $\pi \leftarrow \Perms_N$, $r \leftarrow \bit^*$ and $x \gets \bits^{n-m}$, along with all internal randomness and measurements of $\aD$. Here $\pi^{-1}(y)|_0$ denotes the first bit of $\pi^{-1}(y)$
\item $S=S(\aS)$ is an upper bound on the number of qubits of $\rho$ in the above.
\item $T=T(\aS)$ is an upper bound on the number of oracle queries made by $\aS_1$.
\end{enumerate}
\end{definition}

\begin{definition}\label{def:DPI}
A decision-version permutation inverter ($\dpi$) $\Di=(\Di_0,\Di_1)$ is defined as an $\adpi$ with $m=0$. An $(S,T,\delta)$-$\dpi$ is an $(S,T,\delta)$-$\adpi$ with $m=0$.
\end{definition}

\section{Amplification} \label{sec:amplification}

In this section, we show how to amplify the success probability of search and decision inverters, in the non-adaptive (i.e., $m=0$) case. The construction for the search case is shown in \expref{Protocol}{ptc:eps-SPI repetition}.

\begin{protocol}[$\ell$-time repetition of $(S, T, \epsilon)$-\spi] \label{ptc:eps-SPI repetition} Given an $(S, T, \epsilon)$-\spi $\Si = (\Si_0, \Si_1)$ and an integer $\ell > 0$, define a \spi $\Si[\ell] = (\Si[\ell]_0, \Si[\ell]_1)$ as follows.
\begin{enumerate}
    \item \textit{(Advice Preparation)} $\Si[\ell]_0$ proceeds as follows:
    \begin{enumerate}
        \item receives as input a random permutation $\pi: [N]\rightarrow [N]$ and randomness $r \leftarrow \bit^{*}$ and parses the string $r$ into $2\ell$ substrings $r=r_{0}\Vert...\Vert r_{\ell-1}\Vert r_{\ell}\Vert...\Vert r_{2\ell-1}$ (with lengths as needed for the next step).
        \item uses $r_{0},...,r_{\ell-1}$ to generate $\ell$ permutation pairs $\{\sigma_{1,i},\sigma_{2,i}\}_{i=0}^{\ell-1}$ in $\Perms_N$, and then runs $\Si_0(\sigma_{1,i}\circ \pi \circ \sigma_{2,i}, r_{i+\ell})$ to get a quantum state $\rho_i := \rho_{\sigma_{1,i}\circ \pi \circ \sigma_{2,i}, r_{i+\ell}}$ for all $i \in [0,\ell-1]$. Finally, $\Si[\ell]_0$ outputs the quantum state $\bigotimes_{i=0}^{\ell-1}\rho_i$.
    \end{enumerate}
    \item \textit{(Oracle Algorithm)} $\Si[\ell]_1^{\Opiii}$ proceeds as follows:
    \begin{enumerate}
    \item receives $\bigotimes_{i=0}^{\ell-1}\rho_i$, randomness $r $ and an image $y \in [N]$ as input. 
    \item parses $r = r_{0}\Vert...\Vert r_{\ell-1}\Vert r_{\ell}\Vert...\Vert r_{2\ell-1}$ and uses the coins $r_{0}\Vert...\Vert r_{\ell-1}$ to reconstruct the permutations $\{\sigma_{1,i},\sigma_{2,i}\}_{i=0}^{\ell-1}$ in $\Perms_N$.
    \item runs the following routine for all $i\in[0,\ell-1]$:
    \begin{enumerate}
    \item run $\Si_1$ with oracle access to $(\sigma_{1,i}\circ \pi \circ \sigma_{2,i})_{\bot \sigma_{1,i}(y)}$, which implements the permutation $\sigma_{1,i}\circ \pi \circ \sigma_{2,i}$ and its inverse (with output $\bot$ on input $\sigma_{1,i}(y)$).
\footnote{How to construct this quantum oracle is described in \expref{Appendix}{app:search}.}
      \item get back $x_i \leftarrow \Si_1^{(\sigma_{1,i}\circ \pi \circ \sigma_{2,i})_{\bot \sigma_{1,i}(y)}}(\rho_i,\sigma_{1,i}(y),r_{i+\ell})$. 
    \end{enumerate}
      \item queries the oracle $\Opiii$ (in the forward direction) on each $\sigma_{2,i}(x_i)$ to see if $\pi(\sigma_{2,i}(x_i)) = y$. If such an $\sigma_{2,i}(x_i)$ is found, outputs it; otherwise outputs $0$. 
    \end{enumerate}
\end{enumerate}
\end{protocol}

In the adaptive case, a difficulty arises with the above approach. To amplify the probability, we randomize the permutation in each iteration and $\aS[\ell]_0$ produces corresponding advice for each randomized permutation. In the adaptive case, $\aS[\ell]_0$ needs to output an adaptive string $\mu$ which is used to produce the image $y$. However, running $\aS_0$ for each randomized permutation will, in general, result in a different $\mu$ in each execution, and it is unclear how one can use these to generate a single $\mu'$ in the amplified algorithm.  We remark that other works considered different approaches to amplification, e.g., via quantum rewinding~\cite{hhan2019quantum} and the gentle measurement lemma~\cite{chung2020tight}. 

\begin{lemma}[Amplification, search] \label{lemma:amplify-S}
Let  $\Si$ be a $(S, T, \epsilon)$-\spi for some $\epsilon >0$. Then $\Si[\ell]$ is a $(\ell S, \ell (T+1), 1-(1-\epsilon)^{\ell})$-\spi.
\end{lemma}

\begin{proof}
    We consider the execution of the $\ell$-time repetition of $(S, T, \epsilon)$-\spi, denoted by \textsf{SPI} $\Si[\ell]$, in the search permutation inversion experiment defined in  \expref{Protocol}{ptc:eps-SPI repetition}. By construction, $\Si[\ell]$ runs $\ell$-many \textsf{SPI} procedures $(\Si_0, \Si_1)$. Since $\Si$ is assumed to be an $(S, T, \epsilon)$-\spi, let $\pi_i=\sigma_{1,i}\circ \pi \circ \sigma_{2,i}$, for each iteration $i\in [0,\ell-1]$  it follows that
    \begin{align*}
      &\Pr \left[ (\pi_i)^{-1}(\sigma_{1,i}(y)) \leftarrow
    \Si_1^{(\pi_i)_{\bot \sigma_{1,i}(y)}} \big(\rho_i,\sigma_{1,i}(y),r_{i+\ell} \big) : \rho_i \leftarrow \Si_0(\pi_i, r_{i+\ell})
    \right] \\
    &\equiv \Pr \left[ (\sigma_{2,i})^{-1}\circ \pi^{-1}(y) \leftarrow
    \Si_1^{{(\pi \circ \sigma_{2,i})}_{\bot y}} \big(\rho_{\pi\circ \sigma_{2,i}, r_{i+ \ell}},y,r_{i+\ell} \big): \rho_{\pi\circ \sigma_{2,i}, r_{i+ \ell}} \leftarrow \Si_0(\pi\circ \sigma_{2,i}, r_{r+\ell})
    \right] \\
    &\equiv \Pr \left[ \pi^{-1}(y) \leftarrow
    \Si_1^{\Opiii} \big(\rho_{\pi, r_{i+ \ell}},y,r_{i+\ell} \big): \rho_{\pi, r_{i+ \ell}} \leftarrow \Si_0(\pi, r_{r+\ell})
    \right] \,\,\, \geq \,\,\, \epsilon,
    \end{align*} 
    where the probability is taken over $\pi \leftarrow \Perms_N$ and $r \leftarrow \bit^*$ (which is used to sample permutations $\sigma_i$), along with all internal measurements of $\Si$.

Essentially, for all $i\in [0,\ell-1]$, the goal of the $i$-th trial is to find the preimage $x_i$ such that $\sigma_{2,i}(x_i)=\pi^{-1}(y)$. Since all $\{\sigma_{2,i}\}$ are independently randomly generated, the elements $\sigma_{2,i}(x_i)$ are independent for each $i$ in the range $[0,\ell-1]$. Therefore, all $\ell$ trails are mutually independent. Therefore, we get that 

    \vspace{-5mm}
    \begin{align*}
    &\Pr\left[\pi^{-1}(y) \leftarrow \Si[\ell]_1^{\Opiii}(\rho, y, r): \rho \leftarrow \Si[\ell]_0(\pi, r) \right]\\ 
    &= 1- \Pr \left[ \bigcap_{i=0}^{\ell-1}  \left[ (\pi\circ\sigma_{2,i})^{-1}(y) \not \leftarrow
    \Si_1^{(\pi_i)_{\bot \sigma_{1,i}(y)}}\big(\rho_i,\sigma_{1,i}(y),r_{i+\ell} \big) : \rho_i \leftarrow \Si_0(\pi_i, r_{i+\ell})\right]\right] \\
    &= 1- \prod_{i=0}^{\ell-1} \Pr \left[ (\pi\circ\sigma_{2,i})^{-1}(y) \not \leftarrow
    \Si_1^{(\pi_i)_{\bot \sigma_{1,i}(y)}}\big(\rho_i,\sigma_{1,i}(y),r_{i+\ell} \big) : \rho_i \leftarrow \Si_0(\pi_i, r_{i+\ell})\right] \\
    &\geq 1 - (1 - \epsilon)^\ell.
    \end{align*}
    Given that the \spi $(\Si_0, \Si_1)$ requires space $S$ and $T$ queries, we have that $(\Si[\ell]_0, \Si[\ell]_1)$ requires space $S(\Si[\ell])=\ell \cdot S$ and query number $T(\Si[\ell]) = \ell \cdot (T+1)$, as both algorithms need to run either $\Si_0$ or $\Si_1$ $\ell$-many times as subroutines. This proves the claim.
    \end{proof}

We also need a variant of the above to compute the search lower bound. 
\begin{lemma}\label{lem:amplification}
Let $\Si$ be a $(S, T, \epsilon)$-\spi for some $\epsilon >0$. Then, we can construct an \spi $\Si[\ell] = (\Si[\ell]_0, \Si[\ell]_1)$ using $S(\Si[\ell])$ qubits of advice and making $T(\Si[\ell])$ queries, with
$$
S(\Si[\ell]) = \left\lceil\frac{\ln(10)}{\epsilon}\right\rceil \cdot S \quad \text{ and } \quad T(\Si[\ell]) = \left\lceil\frac{\ln(10)}{\epsilon}\right\rceil \cdot (T+1)
$$
such that
$$
\Pr_{\pi,y} \left[\Pr_r\left[\pi^{-1}(y) \leftarrow \Si[\ell]_1^{\Opiii}(\rho, y, r): \rho \leftarrow \Si[\ell]_0(\pi,r)\right] \geq \frac{2}{3} \right] \geq \frac{1}{5}.
$$
\end{lemma}

The proof is analogous to \expref{Lemma}{lemma:amplify-S} and is given in \expref{Appendix}{app:amp3}.

We also consider amplification for the decision version; the construction is essentially the same, except that the final ``check'' step is replaced by outputting the majority bit. 

\begin{lemma} [Amplification, decision] \label{lemma:amplify-D}
 Let $\Di$ be a $(S,T,\delta)$-\dpi for some $\delta > 0$. Then $\Di[\ell]$ is a $(\ell S, \ell T, 1/2 - \exp(-\delta^2/(1+2\delta) \cdot \ell))$-\dpi.
 \end{lemma}
The proof is analogous to the search version and given in \expref{Appendix}{app:decision}. 

\section{Reductions} \label{sec:reductions}

We give two reductions related to the inversion problem: a search-to-decision reduction (for the case of advice), and a reduction from unstructured search to the decision inversion problem (for the case of no advice).

\subsection{A search-to-decision reduction}

To construct a search inverter from a decision inverter, we take the following approach. We first amplify the decision inverter so that it correctly computes the first bit of the preimage with certainty. We then repeat this amplified inverter $n$ times (once for each bit position) but randomize the instance in such a way that the $j$-th bit of the preimage is permuted to the first position. We then output the string of resulting bits as the candidate preimage. 

\begin{theorem} \label{thm: search-to-decision reduction}
Let $\Di$ be a $(S, T, \delta)$-\dpi. Then for any $\ell \in \N$, we can construct a $(n \ell S, n \ell T, \eta)$-\spi with
$$
\eta \geq 1- n\cdot\exp(-\frac{\delta^2}{(1+2\delta)} \cdot \ell)\,, \quad \text{ where } n=\lceil \log N \rceil.
$$ 
\end{theorem} 

\begin{proof}
  Given an $\delta$-\dpi $(\Di_0, \Di_1)$ with storage size $S$ and query size $T$, we can construct a $\eta'$-\dpi $(\Di[\ell]_0, \Di[\ell]_1)$ with storage size $\ell S$ and query size $\ell T$ through $\ell$-time repetition. By \Cref{lemma:amplify-D}, we have that $\eta' \geq \frac{1}{2}- \exp(-\frac{\delta^2}{(1+2\delta)} \cdot \ell)$. Note that the algorithm $(\Di[\ell]_0, \Di[\ell]_1)$ runs $(\Di_0, \Di_1)$ as a subroutine.
  In the following, we represent elements in $[N]$ using a binary decomposition of length $\lceil \log N \rceil$.
  To state our search-to-decision reduction, we  introduce a generalized swap operation, denoted by $\textsf{swap}_{a,b}$, which acts as follows for any quantum state of $m$ qubits:
  \begin{align*}
      \textsf{swap}_{a,b}\ket{w} &=\textsf{swap}_{a,b}\ket{w_{m-1}\ldots w_b \ldots w_a \ldots w_1w_0}
      = \ket{w_{m-1}\ldots w_a \ldots w_b \ldots w_1w_0}
  \end{align*}
  Note that $\textsf{swap}_{k,k}$ is equal to the identity, i.e. $\textsf{swap}_{k,k}\ket{x}=\ket{x}$ for $x \in [N]$ and $k \in [0,\lceil \log N \rceil-1]$.
  We construct a \spi $(\Si_0, \Si_1)$ as follows.
  \begin{enumerate}
      \item The algorithm $\Si_0$ proceeds as follows:
      \begin{enumerate}
          \item $\Si_0$ receives a random permutation $\pi: [N]\rightarrow [N]$ and a random string $r \leftarrow \bit^{*}$ as inputs. We parse $r$ into $\lceil \log N \rceil$ individual substrings, i.e. $r=r_{0}\Vert...\Vert r_{\lceil \log N \rceil-1}$; the length of each substring is clear in context.
           
          \item $\Si_0$ runs the algorithm ${\Di}[\ell]_0(\pi \circ \textsf{swap}_{0,j}, r_j)$ to obtain 
   quantum advice $\rho_{\pi \circ \textsf{swap}_{0,j}, r_j}$ for each $j \in [0,\lceil \log N \rceil-1]$. Finally, $\Si_0$ outputs a quantum state $\rho = \bigotimes_{j=0}^{\lceil \log N \rceil-1} \rho_{\pi \circ \textsf{swap}_{0,j}, r_j}$. (Note: We let $\rho_j=\rho_{\pi \circ \textsf{swap}_{0,j}, r_j}$ for the rest of the proof.)
      \end{enumerate}
      \item The oracle algorithm $\Si_1^{\algo O_\pi,\Opii}$ proceeds as follows:\footnote{Here, we borrow the notation for $\algo O_{\pi}$ and $\Opii$ from the experiment described in \expref{Section}{sec:perminv}.} 
      \begin{enumerate}
      \item $\Si_1$ receives $\bigotimes_{j=0}^{n-1} \rho_j$, a random string $r:=r_{0}\Vert...\Vert r_{n-1}$ and an image $y \in [N]$.
      \item  $\Si_1$ then runs the following routine for each $j\in[0,\lceil \log N \rceil-1]$:
      \begin{enumerate}
      \item Run $\Di[\ell]_1$ with oracle access to $\algo O_{\pi \circ \swap_{0,j}}$ and $\algo O_{(\pi \circ \swap_{0,j})^{-1}_{\bot y}}$, where
      \begin{align*}
          \algo O_{\pi \circ \swap_{0,j}}(\ket{w}_1\ket{z}_2)&=\left(\swap_{0,j}\otimes I \right)\algo O_\pi \left(\swap_{0,j}\otimes I \right)\ket{w}_1\ket{z}_2 \\
           \algo O_{(\pi \circ \swap_{0,j})^{-1}_{\bot y}}(\ket{w}_1\ket{z}_2)&= (  I\otimes \swap_{0,j})\Opii \ket{w}_1\ket{z}_2
      \end{align*}
      \item Let $b_j \leftarrow {\Di}[\ell]_1^{(\pi \circ \swap_{0,j})_{\bot y}}(\rho_j,y,r_j)$ denote the output.
      \end{enumerate}
    \item $\Si_1$ outputs $x^* \in [N]$ with the binary decomposition $ x^* = \sum_{j=0}^{\lceil \log N \rceil -1} 2^j \cdot b_{j}$.
      \end{enumerate}
  \end{enumerate}
  We now argue that the probability that $\D[\ell]_1$ correctly recovers the pre-image bits $b_i$ and $b_j$ is independent for each $i \neq j$.
  From \expref{Lemma}{lemma:amplify-D}, we know that $\Di[\ell]_1$ runs $\Di_1$ as a subroutine, i.e. it decides the first bit of the pre-image of $y$ by running $\Di_1$ (in \expref{Lemma}{lemma:amplify-D}) $\ell$ times with different random coins. It actually needs to recall $\Di_1$ for amplification and for each iteration in this amplification $k \in [0, \ell -1 ]$, where the actual modified permutation under use is $\sigma_{i,k} \circ \pi \circ \textsf{swap}_{0,i}$ and image is $\sigma_{i,k}(y)$. Similarly for term $j$, $\sigma_{j,k} \circ \pi \circ \textsf{swap}_{0,j}$ and $\sigma_{j,k}(y)$ is used as the  permutation and image. Since the random coins ($r_i$ and $r_j$), which are used to modify the target permutation $\pi$, are independently random, those random permutations ($\sigma_{i,k}$ and $\sigma_{j,k}$) generated from random coins are independently random and so do those modified composition permutations, images, and advice states.
  
  Analyzing the success probability of $(\Si_0, \Si_1)$, we find that
  \begin{align*}
  &\textsf{Pr}\left[\pi^{-1}(y) \leftarrow \Si_1^{\Opiii}(\rho, y, r) \,\, : \,\, \rho \leftarrow \Si_0(\pi, r) \right]\\
  &=\textsf{Pr} \left[   \bigwedge_{j=0}^{\lceil \log N \rceil -1}
     \pi^{-1}(y)|_j
     \leftarrow \Di[\ell]_1^{(\pi \circ \swap_{0,j})_{\bot y}}(\rho_{j},y,r_{j})
     \right]\\
      &\geq \left(1- \exp(-\frac{\delta^2}{(1+2\delta)} \cdot \ell)\right)^{\lceil \log N \rceil}
      \geq 1- \lceil \log N \rceil \cdot \exp(-\frac{\delta^2}{(1+2\delta)} \cdot \ell).
  \end{align*}
  where the last line follows from Bernoulli's inequality.
  Finally, we compute the resources needed for $(\Si_0, \Si_1)$. By \expref{Lemma}{lemma:amplify-D}, $(\Di[\ell]_0, \Di[\ell]_1)$ requires space $\ell S$ and query size $\ell T$. For $j \in [0,\lceil \log N \rceil -1]$, $\Si_0$ stores $\Di[\ell]_0$'s outputs and thus $\Si$ requires storage size $\lceil \log N \rceil  \ell S$. Similarly, $\Si_1$ runs $\Di[\ell]_1$ to obtain $b_j$ and thus it requires $\lceil \log N \rceil  \ell T$ queries in total.
    \end{proof}

\paragraph{Comparison to O2H lemma.} The one-way to hiding (O2H) lemma~\cite{ambainis2019quantum} also presents a natural reduction from search to decision in the context of general quantum oracle algorithms. However, it is quite limited in our setting. For example, given a decision inverter capable of computing the first bit of $\pi^{-1}(y)$ with certainty after $q$ queries, the O2H lemma yields a search inverter that can invert $y$ with success probability $\frac{1}{4q^2}$ after $\approx q$ queries. By comparison, our amplification technique achieves an inversion of $y$ with a success probability of $1$ with $nq$ queries, which is significantly better in the relevant setting of $q \gg n$. However, in applications where only one copy of the advice is available for the amplified algorithm, O2H still works while our amplification technique fails.

\subsection{A reduction from unstructured search}

Second, we generalize the method used in \cite{nayak2010inverting} to give a lower bound for adaptive decision inversion without advice. Unlike in Nayak's original reduction, here we grant two-way access to the permutation. Recall that, in the unique search problem, one is granted quantum oracle access to a function $f: [N] \rightarrow  \bool$ which is promised to satisfy either $|f^{-1}(1)| = 0$ or $|f^{-1}(1)| = 1$; the goal is to decide which is the case. The problem is formally defined below. 

\begin{definition} ($\textsf{UNIQUESEARCH}_n$)
  Given a function $f: \{0,1\}^n \rightarrow  \bool$, such that $f$ maps at most one element to 1, output YES if $f^{-1}(1)$ is non-empty and NO otherwise. 
  \end{definition}
  \begin{definition}(Distributional error)
  Suppose an algorithm solves a decision problem with error probability at most $p_0$ for NO instances and $p_1$ for YES instances. Then we say this algorithm has distributional error $(p_0,p_1)$. 
  \end{definition}

We now establish a reduction from unstructured search to adaptive decision inversion.

\begin{theorem} \label{thm:Nayak}
If there exists a $(0, T, \delta)$-\adpi, then there exists a quantum algorithm that solves $\textsf{UNIQUESEARCH}_{n-m-1}$ with at most $2T$ queries and distributional error $\left(\frac{1}{2}-\delta,\frac{1}{2}\right)$.
\end{theorem}
\begin{proof}
  Our proof is similar to that of Nayak~\cite{nayak2010inverting}: given a $(0, T, \delta)$-\adpi $\algo A$, we construct another algorithm $\algo B$ which solves the $\textsf{UNIQUESEARCH}_{n-m-1}$ problem. 
  
  Let $N =2^n$. For any uniform image $t\in [N]$, define the NO and YES instances sets (corresponding to the image $t$) of the decision permutation inversion problem with size $N$: 

    $$
    \pi_{t,0}=\{\pi: \pi \text{ is a permutation on } [N], \text{ the first bit of } \pi^{-1}(t) \text{ is } 0\}, 
    $$
    $$
    \pi_{t,1}=\{\pi: \pi \text{ is a permutation on } [N], \text{ the first bit of } \pi^{-1}(t) \text{ is } 1\}.
    $$
  Note that for a random permutation $\pi$, whether $\pi\in \pi_{t,0}$ or $\pi_{t,1}$ simply depends on the choice of $t$. Since $t$ is uniform, $\Pr[\pi\in \pi_{t,0}]=\Pr[\pi\in \pi_{t,1}]=1/2$. 
  We also consider functions $h: [N] \rightarrow [N]$ with a unique collision at $t$. One of the colliding pairs should have the first bit $0$, and the other one should have the first bit $1$. Moreover, the last $m$ bits of the colloding pair is $\mu$. Formally speaking, $h(0\|i\|\mu)=h(1\|j\|\mu)=t$, where $i,j  \in \{0,1\}^{n-m-1}$. Let $Q_{t,\mu}$ denote the set of all such functions. 

  Furthermore, given a permutation $\pi$ on $[N]$, consider functions in $Q_{t,\mu}$ that differ from $\pi$ at exactly one point. These are functions $h$ with a unique collision and the collision is at $t$. If $\pi\in \pi_{t,0}$, $\pi(0\|i\|\mu)=h(0\|i\|\mu)=t$ and $1\|j\|\mu$ is the unique point where $\pi$ and $h$ differ; if $\pi\in \pi_{t,1}$, $\pi(1\|j\|\mu)=h(1\|j\|\mu)=t$ and $0\|i\|\mu$ is the unique point where $\pi$ and $h$ differ. Let $Q_{\pi,t,\mu}$ denote the set of such functions $h$ and clearly $Q_{\pi,t,\mu} \subseteq Q_{t,\mu}$. Note that if we pick a random permutation $\pi$ in $\{\pi_N\}$ and choose a uniform random $h\in Q_{\pi,t,\mu}$, $h$ is also uniform in $Q_{t,\mu}$. 
  Next, we construct an algorithm $\algo B$ that tries to solve $\textsf{UNIQUESEARCH}_{n-m-1}$ as follows, with quantum oracle access to $f$:
  \begin{enumerate}
      \item $\algo B$ first samples some randomness $r\in \bool^{*}$, a uniform random string $s \in \{0,1\}^{n-m}$ and a permutation $\pi \in \{\pi_N\}$. 
      \item $\algo B$ then runs $\algo A$ with quantum orale access to $\pi,\pi^{-1}$ until it receives a string $\mu \in \bool^{m}$ from $\algo A$. 
      \item Let $t = \pi(s\|\mu)$, and then it follows that if $s|_0 = 0$, $\pi \in \pi_{t,0}$, and otherwise $\pi \in \pi_{t,1}$. 
      \item $\algo B$ then constructs a function $ h_{f,\pi,t,\mu}$ and $h_{f,\pi,t,\mu}^{-1*}$ as follows. If $\pi \in \pi_{t,0}$, for any $i\in \{0,1\}$ and $j\in \bool^{n-m-1}$,
          \begin{equation}
          h_{f,\pi,t, \mu}(i\|j\|u)=
          \begin{cases}
          t & \text{ if }   i=1 \text{ and } f(j)=1, u=\mu, \\
          \pi(i\|j\|u) & \text{ otherwise. } 
          \end{cases}    
          \end{equation}
          If $\pi \in \pi_{t,1}$, for any $i\in \{0,1\}$ and $j\in \bool^{n-m-1}$,
          \begin{equation}
          h_{f,\pi,t,\mu}(i\|j\|\mu)=
          \begin{cases}
          t & \text{ if }   i=0 \text{ and } f(j)=1,  u=\mu, \\
          \pi(i\|j\|u) & \text{ otherwise. } 
          \end{cases}
          \end{equation}
          No matter what instance sets $\pi$ belongs to, the corresponding "inverse" function is defined as 
          \begin{equation}
          h_{f, \pi, t,\mu}^{-1*}(k||b)=
          \begin{cases}
          \pi^{-1}(k)\|0 & \text{ if }   b=0 \text{ and } k \neq t,  \\
          1\|1 & \text{otherwise.}
          \end{cases}
          \end{equation}
         
      \item $\algo B$ then sends $t$, $\mu$ and $r$ to $\algo A$, runs it with quantum oracle access to $ h_{f,\pi,t,\mu}$ and $ h_{f,\pi,t,\mu}^{-1*}$, and finally gets back $b'$. For simplicity, we write this process as $b' \leftarrow \algo A^{h_{\bot t}}(t,\mu,r)$. \footnote{Note that those functions are defined classically above, and its allowance for quantum oracle access is discussed in \Cref{app:reduction_Nayak}, which gives $2q$ queries in the theorem statement.}
      \item $\algo B$ outputs $b'$ if $\pi \in \pi_{t,0}$, and $1-b'$ if $\pi \in \pi_{t,1}$.
  \end{enumerate}
   Let $\delta_1$ be the error probability of $\algo A$ in the YES case and $\delta_0$ be that in the NO case of $(0, T, \delta)$-\dpi. Since $s$ is uniform random and then $\Pr[\pi\in \pi_{t,0}]=\Pr[\pi\in \pi_{t,1}]=1/2$, it follows that
   $$\Pr[\text{error of } \algo A] =1-\left(\frac{1}{2}+\delta \right)=\frac{1}{2}(\delta_0+\delta_1) \Rightarrow\delta=\frac{1}{2}-\frac{1}{2}(\delta_0+\delta_1).$$ 
  
   We now analyze the error probability of $\algo B$ in the YES and NO cases. In the NO case, $f^{-1}(1)$ is empty, so no matter whether $\pi\in \pi_{t,0}$ or $\pi\in \pi_{t,1}$, $h_{f,\pi,t,\mu}=\pi$. It follows that $\algo A^{h_{\bot t}}(t,r)= \algo A^{\pi_{\bot t}}(t,r)$. Therefore,
   \begin{align*}
  \Pr[\text{error of $\algo B$ in NO case}] &= \Pr[1 \leftarrow \algo B^{\algo O_f}(\cdot)]\\
       &=  \Pr[1 \leftarrow \algo A^{h_{\bot t}}(t,r) | \pi \in \pi_{t,0}] \Pr[ \pi \in \pi_{t,0}] \\
       & \ \ \ \ \ +  \Pr[0 \leftarrow \algo A^{h_{\bot t}}(t,r) | \pi \in \pi_{t,1}] \Pr[ \pi \in \pi_{t,1}] \\
       &=  \frac{1}{2} \left(\Pr[1 \leftarrow \algo A^{\pi_{\bot t}}(t,r) | \pi \in \pi_{t,0}]   +  \Pr[0 \leftarrow \algo A^{\pi_{\bot t}}(t,r) | \pi \in \pi_{t,1}] \right)\\
       &=\frac{1}{2} \left(\Pr[\text{error of $\algo A$ in NO case}] + \Pr[\text{error of $\algo A$ in YES case}] \right) \\
       &= \frac{1}{2} \left(\delta_0 + \delta_1\right) = \frac{1}{2} - \delta.
   \end{align*}

  In the YES case, $f^{-1}(1)$ is not empty, so function $h_{f,\pi,t,\mu}$ has a unique collision at $t$, with one of the colliding pair having first bit $0$ and the other one having first bit $1$, no matter $\pi\in \pi_{t,0}$ or $\pi_{t,1}$. As $f$ is a black-box function, the place $j$ where $f(j)=1$ is uniform and so $h_{f,\pi,t, \mu}$ is uniform in $Q_{\pi,t,\mu}$. 
  By arguments at the beginning of this proof, as $\pi$ is uniform, the function is also uniform in $Q_{t,\mu}$. 
  Let $p:=\Pr \limits_{h_{f,\pi,t,\mu}\leftarrow Q_{t,\mu}}[ 0 \leftarrow \algo A^{h_{\bot t}}(t,r)].$ Therefore, 
   \begin{align*}
  \Pr[\text{error of $\algo B$ in YES case}]
       &= \Pr[0 \leftarrow \algo B^{f}(\cdot)]\\
       &=  \Pr[0 \leftarrow \algo A^{h_{\bot t}}(t,r) | \pi \in \pi_{t,0}] \Pr[ \pi \in \pi_{t,0}] \\
       & \ \ \ \ \ +  \Pr[1 \leftarrow \algo A^{h_{\bot t}}(t,r) | \pi \in \pi_{t,1}] \Pr[ \pi \in \pi_{t,1}] \\
       &= \frac{1}{2} \left( \Pr[0 \leftarrow \algo A^{h_{\bot t}}(t,r) |h_{f,\pi,t,\mu} \xleftarrow{\$} Q_{t,\mu}] \right.\\
       & \ \ \ \ \ + \left. \Pr[1 \leftarrow \algo A^{h_{\bot t}}(t,r) | h_{f,\pi,t, \mu}\xleftarrow{\$} Q_{t,\mu}] \right)\\
       &=\frac{1}{2} \left(p + (1-p) \right) = \frac{1}{2}.
   \end{align*}
       where the third equality comes from the fact stated above: no matter $\pi \in \pi_{t,0}$ or $\pi \in \pi_{t,1}$, the corresponding $h$ is uniform in $Q_{t,\mu}$ and then can be viewed as uniform randomly generated from $Q_{t,\mu}$. Since $\A$ is granted with oracle access to $h$, both conditions can be changed to $h_{f,\pi,t,\mu} \xleftarrow{\$} Q_{t,\mu}$.
       Note that given $h$, even if $\A$ can notice that it is not a permutation and then acts arbitrarily, this can only influence the probability of two terms individually, i.e. the value of $p$ and $1-p$. But as we only care about their summation, we do not need to handle the consequence of $\A$ noticing the difference, including the probability of oracle distinguishability.
    \end{proof}

\section{Lower bounds} \label{sec:lower bound}

\subsection{Search version} \label{sec:search-lower-bound}

We now give lower bounds for the search version of the permutation inversion problem over $[N]$. We begin with a lower bound for a restricted class of inverters (and its formal definition); these inverters succeed on an $\epsilon$-fraction of inputs with constant probability (say, $2/3$.). The proof uses a similar approach as in previous works on one-sided permutation inversion with advice~\cite{nayebi2014quantum, chung2019lower,hhan2019quantum}. 

\begin{theorem} \label{thm:lower bounds for perm inversion}
Let $N \in \N$. Let $\Si = (\Si_0, \Si_1)$ be a $(S, T, 2\epsilon/3)$-\spi that satisfies
$$
\Pr_{\pi,y} \left[\Pr_r\left[\pi^{-1}(y) \leftarrow \Si_1^{\Opiii}(\rho, y, r): \rho \leftarrow \Si_0(\pi,r)\right] \geq \frac{2}{3} \right] \geq \epsilon.
$$ We call those inverters \textit{restricted inverters}.
Suppose that $\epsilon = \omega(1/N)$, $T= o(\epsilon \sqrt{N})$ and $S\geq 1.$ Then, for sufficiently large $N$ we have $ST^2 \geq \widetilde{\Omega}(\epsilon N)$.
\end{theorem}

\begin{proof}
    To prove the claim, we construct a \QRACVL scheme that encodes the function $\pi^{-1}$ and then derive the desired space-time trade-off via \Cref{cor:lower-bound-QRAC-VL}. Let  $\Si = (\Si_0, \Si_1)$ be an $2\epsilon/3$-\spi that succeeds on a $\epsilon$-fraction of inputs with probability at least $2/3$. In other words, $\Si$ satisfies
    $$
    \Pr_{\pi,y} \left[\Pr_r\left[\pi^{-1}(y) \leftarrow \Si_1^{\Opiii}(\rho, y, r): \rho \leftarrow \Si_0(\pi, r)\right] \geq \frac{2}{3} \right] \geq \epsilon.
    $$ 
    By the averaging argument in \expref{Lemma}{lem:averaging} with parameter $\theta = 1/2$, it follows that there exists a large subset $\algo X \subseteq \algo S_N$ of permutations with size at least $N!/2$ such that for any permutation $\pi \in \algo X$, we have that
    $$
    \Pr_{y} \left[\Pr_r\left[\pi^{-1}(y) \leftarrow \Si_1^{\Opiii}(\rho, y, r): \rho \leftarrow \Si_0(\pi, r)\right] \geq \frac{2}{3} \right]  \geq \frac{\epsilon}{2}.
    $$
    For a given permutation $\pi \in \algo X$ we let $\algo I$ be the set of indices $x \in [N]$ such that $\Si$ correctly inverts $\pi(x)$ with probability at least $2/3$ over the choice of $r$. By the definition of the set $\algo X$, we have that $|\algo I|\geq \epsilon/2 \cdot N$. 
    Our \QRACVL scheme $(\Enc,\Dec)$ for encoding permutations is described in detail in \expref{Protocol}{ptc:QRAC-perm}. Below, we introduce some additional notations that will be relevant to the scheme. 
    For convenience, we model the two-way accessible oracle given to $\Si_1$ in terms of a single oracle for the \emph{merged} function of the form \footnote{The (reversible) quantum oracle implementation is similar to the one in \expref{Definition}{def:SPI}. We use the function $\Opiii$ for ease of presentation since the same proof carries over with minor modifications in the quantum oracle case.}
    $$
    \Opiii(w,a) \overset{\text{def}}{=} \begin{cases}
    \pi(w) & \text{ if } a=0\\
    \pi^{-1}(w) & \text{ if } w \neq y  \, \land \, a=1\\
    \bot & \text{ if } w = y \, \land \, a=1.
    \end{cases}
    $$
    Let $c,\gamma \in (0,1)$ be parameters.
    As part of the encoding, we use the shared randomness $R \in \bit^*$ to sample a subset $\algo R \subseteq [N]$ such that each element of $[N]$ is contained in $\algo R$ with probability $\gamma/T(\Si)^2$.
    Moreover, we define the following two disjoint subsets of $[N] \times \bit$:
    \begin{align*}
    \Sigma_0^{\algo R} &= \algo R\setminus \{x\} \times \{0\}\\
    \Sigma_1^{\algo R} &= \pi(\algo R) \setminus \{\pi(x)\} \times \{1\}.
    \end{align*}
    Let $\algo G \subseteq \algo I$ be the set of $x \in [N]$ which satisfy the following two properties: 
    \begin{enumerate}
        \item 
        The element $x$ is contained in the set $\algo R$, i.e.
        \begin{align} \label{eq:x-in-R}
            x \in \algo R;
        \end{align} 
        \item The total query magnitude of $\Si_1^{\Opiii}$ with input $(\Si_0(\pi,r),y,r)$ on the set $\Sigma_0^{\algo R} \cup \Sigma_1^{\algo R}$
        is bounded by $c/T(\Si)$. In other words, we have
        \begin{align} \label{eq:TQM}
        q(\algo \Si_1^{\Opiii},\Sigma_0^{\algo R} \cup \Sigma_1^{\algo R})\,\leq \, c/T(\Si).
        \end{align}
    \end{enumerate}
    \textbf{Claim 1.} 
    Let $\algo G \subseteq [N]$ be the set of $x$ which satisfy the conditions in~\eqref{eq:x-in-R} and~\eqref{eq:TQM}. Then, there exist constants $\gamma,c \in (0,1)$ such that
    $$
    \Pr_{\algo R} \left[ |\algo G| \geq \frac{\epsilon\gamma N}{4 \,T(\Si)^2} \left(1 - \frac{5 \gamma^2}{c}\right) \right] \geq 0.8.
    $$
    In other words, we have $|\algo G| = \Omega(\epsilon N / T(\Si)^2)$ with high probability.
    \begin{proof} (of the claim)
    Let $\algo H = \algo R \cap \algo I$ denote the set of $x \in \algo R$ for which $\Si$ correctly inverts $\pi(x)$ with probability at least $2/3$ over the choice of $r$. By the definition of the set $\algo R$, it follows that $|\algo H|$ has a binomial distribution. Therefore, in expectation, we have that $|\algo H|=\gamma|\algo I|/T(\Si)^2$. Using the multiplicative Chernoff bound in \expref{Lemma}{lma:chernoff} and the fact that $T(\Si)=o(\epsilon \sqrt{N})$, we get
    \begin{align}\label{eq:H}
      \Pr_{\algo R} \left[|\algo H| \geq \frac{\gamma |\algo I|}{2 \, T(\Si)^2} \right] \geq 0.9, 
    \end{align}
    for all sufficiently large $N$.
    Because each query made by $\Si_1$ has unit length and because $\Si_1$ makes at most $T(\Si)$ queries, it follows that
    \begin{equation}\label{eq:bound-T}
    q(\algo \Si_1^{\Opiii},[N] \times \bit) \leq T(\Si).
    \end{equation}
    We obtain the following upper bound for the average total query magnitude:
    \begin{align*} &\underset{\algo R}{\mathbb{E}} \left[ q(\algo \Si_1^{\Opiii},\Sigma_0^{\algo R} \cup \Sigma_1^{\algo R})
    \right]\\
    &=
    \underset{\algo R}{\mathbb{E}} \left[ q(\algo \Si_1^{\Opiii},\Sigma_0^{\algo R}) + q(\algo \Si_1^{\Opiii},\Sigma_1^{\algo R})
    \right] \quad\quad\quad\quad\quad \quad\quad\quad\quad(\Sigma_0^{\algo R},\Sigma_1^{\algo R}\text{ are disjoint})
    \\
    &= \underset{\algo R}{\mathbb{E}} \left[ q(\algo \Si_1^{\Opiii},\Sigma_0^{\algo R})\right]+ \underset{\algo R}{\mathbb{E}} \left[q(\algo \Si_1^{\Opiii},\Sigma_1^{\algo R})
    \right] \quad\quad\quad\quad\quad\quad\quad (\text{linearity of expectation})\\
    &= \underset{\algo R}{\mathbb{E}} \left[ q(\algo \Si_1^{\Opiii},\algo R\setminus \{x\} \times \{0\})\right]+ \underset{\algo R}{\mathbb{E}} \left[q(\algo \Si_1^{\Opiii},\pi(\algo R)\setminus \{\pi(x)\} \times \{1\})
    \right]\\
    &= \frac{\gamma}{T(\Si)^2} \cdot q(\algo \Si_1^{\Opiii},[N]\setminus \{x\} \times \{0\}) + \frac{\gamma}{T(\Si)^2} \cdot q(\algo \Si_1^{\Opiii},\pi([N])\setminus \{\pi(x)\} \times \{1\})\\
    &= \frac{\gamma}{T(\Si)^2} \cdot q(\algo \Si_1^{\Opiii},[N]\setminus \{x\} \times \{0\})\\
    & \quad + \frac{\gamma}{T(\Si)^2} \cdot q(\algo \Si_1^{\Opiii},
    [N]\setminus \{\pi(x)\} \times \{1\})
    \quad\quad\quad\quad\quad\quad\quad\quad (\text{$\pi$ is a permutation})\\
    &\leq \frac{\gamma}{T(\Si)^2} \cdot \big[q(\algo \Si_1^{\Opiii},[N] \times \{0\}) + q(\algo \Si_1^{\Opiii},[N]\times \{1\})\big] \quad\quad\quad\quad(\text{supersets})\\
    &=\frac{\gamma}{T(\Si)^2} \cdot q(\algo \Si_1^{\Opiii},[N]\times \{0,1\}) \,\,\leq \,\, \frac{\gamma}{T(\Si)}. \quad\quad\quad\quad\quad\quad(\text{by the inequality in \eqref{eq:bound-T}})
    \end{align*}
    Hence, by Markov's inequality,
    \begin{align}\label{eq:Markov}
    \underset{\algo R}{\Pr} \left[ 
    q(\algo \Si_1^{\Opiii},\Sigma_0^{\algo R} \cup \Sigma_1^{\algo R})\geq \frac{c}{T(\Si)}
    \right] \leq \frac{T(\Si)}{c} \cdot \frac{\gamma}{T(\Si)} = \frac{\gamma}{c}.
    \end{align}
    
    Let us now denote by $\algo J$ the subset of $x \in \algo I$ that satisfy Eq.~\eqref{eq:x-in-R} but not Eq.~\eqref{eq:TQM}. Note that Eq.~\eqref{eq:x-in-R} and Eq.~\eqref{eq:TQM} are independent for each $x \in \algo I$, since Eq.~\eqref{eq:x-in-R} is about whether $x \in \algo R$ and Eq.~\eqref{eq:TQM} only concerns the intersection of $\algo R$ and $[N] \setminus \{x\}$, as well as $\pi(\algo R)$ and $\pi([N]) \setminus \{\pi(x)\}$. Therefore, by \eqref{eq:Markov}, the probability that $x \in \algo I$ satisfies $x \in \algo J$ is at most $\gamma^2/(cT(\Si)^2)$. Hence, by Markov's inequality,
    \begin{align}\label{eq:J}
    \underset{\algo R}{\Pr} \left[
    |\algo J| \leq \frac{10 |\algo I| \gamma^2}{cT(\Si)^2}\right] \geq 0.9.    
    \end{align}
    Using \eqref{eq:H} and \eqref{eq:J}, we get with probability at least $0.8$ over the the choice of $\algo R$,
    $$
    |\algo G| = |\algo H| - |\algo J| \geq \frac{|\algo I|\gamma}{2T(\Si)^2} - \frac{10 |\algo I| \gamma^2}{cT(\Si)^2} \geq \frac{\epsilon\gamma N}{4T(\Si)^2} \left(1 - \frac{5 \gamma^2}{c}\right),
    $$
    given that $\gamma$ is a sufficiently small positive constant.
      \end{proof}

    \begin{protocol}[Quantum Random Access Code For Inverting Permutations] \label{ptc:QRAC-perm}\ \\
    Let $c,\gamma \in (0,1)$ be parameters.
    Consider the following (variable-length) quantum random-access code given by $\QRACVL = (\Enc,\Dec)$ defined as follows:
    \begin{description}
    \item[$\bullet$] $\Enc(\pi^{-1}; R)$: On input $\pi^{-1} \in \algo S_N$ and randomness $R \in \bit^*$, first
    use $R$ to extract random coins $r$ and then proceed as follows:
    \begin{description}
    \item[Case 1:] $\pi \notin \algo X$ or $|\algo G| < \frac{\epsilon\gamma N}{4 \,T(\Si)^2} \left(1 - \frac{5 \gamma^2}{c}\right)$.
    Use the classical flag $\texttt{case}=1$ (taking one additional bit) and output the entire permutation table of $\pi^{-1}$.
    \item[Case 2:] $|\algo G| \geq \frac{\epsilon\gamma N}{4 \,T(\Si)^2} \left(1 - \frac{5 \gamma^2}{c}\right)$.
    Use the classical flag $\texttt{case}=2$ (taking one additional bit) and output the following
    \begin{enumerate}
     \item The size of $\algo G$, encoded using $\log N$ bits;
        \item the set $\algo G \subseteq \algo R$, encoded using $\log\binom{|\algo R|}{|\algo G|} $ bits;
        \item The permutation $\pi$ restricted to inputs outside of $\algo G$, encoded using \\ $\log(N!/|\algo G|!)$ bits;
        \item Quantum advice used by the algorithm repeated $\rho$ times with $\alpha^{\otimes \rho}$, for $\alpha \leftarrow \Si_0(\pi, r)$ for some $\rho$ that we will decide later. (We can compute this as the encoder can preprocess multiple copies of the same advice. Note that this
    is the only part of our encoding that is not classical.)
    \end{enumerate}
    \end{description}
    \item[$\bullet$] $\Dec(\beta, y; R)$: On input encoding $\beta$, image $y \in [N]$ and randomness $R \in \bit^*$, first
    use $R$ to extract random coins $r$ and then proceed as follows:
    \begin{description}
    \item[Case 1:] This corresponds to the flag $\texttt{case}=1$.
    Search the permutation table for $\pi^{-1}$ and output $x$ such that $\pi^{-1}(y)=x$.
    \item[Case 2:] This corresponds to the flag $\texttt{case}=2$.
    Recover $\algo G$ and $\pi(x)$ for every $x \notin \algo G$. If $y = \pi(x)$ for some $x \notin \algo G$, output $x = \pi^{-1}(y)$. Otherwise, parse $\alpha_1,\alpha_2,\dots,\alpha_\rho$ and run $\Si_1^{\bar{\pi}_{\bot y}}(\alpha_i,y,r)$ for each $i \in [\rho]$ and output their majority vote, where we let \footnote{The (reversible) quantum oracle implementation for $\bar{\pi}_{\bot y}$ is provided in \Cref{sec:oracle-for-pi-bar}.}
    $$
    \bar{\pi}_{\bot y}(w,a) = \begin{cases}
    y & \text{ if } w \in \algo G \,\land\, a=0\\
    \pi(w) & \text{ if } w \notin \algo G \,\land\, a=0\\
    \pi^{-1}(w) & \text{ if } w \notin \pi(\algo G)  \, \land \, a=1\\
    \bot & \text{ if } w \in \pi(\algo G) \, \land \, a=1.
    \end{cases}
    $$
    \end{description}
    \end{description}
    \end{protocol}
    Let us now analyze the performance of our \QRACVL scheme $(\Enc,\Dec)$ in \expref{Protocol}{ptc:QRAC-perm}. 
     Let $\ket{\Psi_{\Opiii}}$ and $\ket{\Psi_{\bar{\pi}_{\bot y}}}$ denote the final states of $\Si_1$ when it is given the oracles $\Opiii$ and $\bar{\pi}_{\bot y}$, respectively. By \expref{Lemma}{lem:swapping} and the properties of the total query magnitude:
    \begin{align*}
    \|\ket{\Psi_{\Opiii}} -  \ket{\Psi_{\bar{\pi}_{\bot y}}} \| &\leq \sqrt{T(\Si) \cdot q(\algo \Si_1^{\Opiii}, \algo G\setminus \{x\} \times \{0\}) \,\cup \, (
        \pi(\algo G) \setminus \{\pi(x)\} \times \{1\}) }\\
    &\leq \sqrt{T(\Si)  \cdot q(\algo \Si_1^{\Opiii},\Sigma_0^{\algo R} \cup \Sigma_1^{\algo R}) } \\
        &\leq \sqrt{T(\Si) \cdot \frac{c}{T(\Si)}} = \sqrt{c}.
    \end{align*}

    Since $x \in \algo I$, it follows from the definition of $\algo I$ that measuring $\ket{\Psi_{{\pi}_{\bot y}}}$ results in $x$ with probability at least $2/3$. Given a small enough positive constant $c$, we can ensure that measuring $\ket{\Psi_{\bar{\pi}_{\bot y}}}$ will result in $x$ with probability at least $0.6$.
    We now examine the length of our encoding. With probability $1- \epsilon/2$, we have $\pi \notin \algo X$; with probability $\epsilon(1-0.8)/2$, we have $\pi \in \algo X$ but $\algo G$ is small, i.e.,
    $$
    |\algo G| < \frac{\epsilon\gamma N}{4 \,T(\Si)^2} \left(1 - \frac{5 \gamma^2}{c}\right).
    $$
    Therefore, except with probability $1- 0.4\epsilon$, our encoding will result in the flag $\texttt{case}=1$, where the encoding consists of $1 + \log N!$ classical bits and the decoder succeeds with probability $1$.
    With probability $0.4\epsilon$ , our encoding has the flag $\texttt{case}=2$, and the size equals
    $$
    1 + \log N + \log\binom{|\algo R|}{|\algo G|} + \log(N! / |\algo G|!) + \rho S(\Si).
    $$
    By the assumption that $T(\Si)= o(\epsilon \sqrt{N})$, we have
    \begin{align*}
        \log\binom{|\algo R|}{|\algo G|} &= \log \left(\frac{|\algo R|(|\algo R|-1)\ldots (|\algo R|-|\algo G| +1)}{|\algo G|(|\algo G|-1)\ldots 1}\right)\\
        & = O\left(\log \left(\frac{|\algo R||\algo R|\ldots |\algo R|}{|\algo G||\algo G|\ldots |\algo G|}\right)\right)\\
        & = O(|\algo G|\log(|\algo R|/|\algo G|))\\
        &= O(|\algo G|\log 1/\epsilon) \\
        &= o(|\algo G| \log |\algo G|),
    \end{align*}
    and we can rewrite the size of the encoding as
    $$
    \log N + o(|\algo G|\log |\algo G|) + \log N! - \log |\algo G|!  + \rho S(\Si).
    $$
    
    In the case when the decoder is queried on an input that is already known, that is $y \notin \pi(\algo G)$ (which occurs with probability $1 - |\algo G|/N$), the decoder recovers the correct pre-image with probability $1$. Otherwise, the analysis is the following: with just one copy of the advice, the decoder recovers the correct pre-image with probability $2/3$, and hence with $\rho$ many copies, the decoder can take the majority vote and recover the correct pre-image with probability $1 - \exp(-\Omega(\rho))$. The latter follows from the Chernoff bound in \Cref{lma:chernoff}. 
    Overall, the average encoding length is 
    $$
    0.4\epsilon \cdot (\log N + o(|\algo G|\log |\algo G|) - \log |\algo G|!  + \rho S(\Si)) + \log N!$$ 
    
    where 
    the average success probability is $1 - |\algo G|/N \cdot \exp(-\Omega(\rho))$. By setting $\rho = \Omega(\log(N/\epsilon)) = \Omega(\log N)$, the
    average success probability amounts to
    $1 - O(1/N^2)$. Therefore, using the lower bound in \Cref{cor:lower-bound-QRAC-VL}, we have 
    \begin{align*}
    \log N! + 0.4\epsilon \cdot (\log N + o(|\algo G|\log |\algo G|) - \log |\algo G|!  + \rho S(\Si))  &\geq \log N! - O\left(\frac{1}{N}\log N\right)\\
    \log N + o(|\algo G|\log |\algo G|) - \log |\algo G|!  + \rho S(\Si)  &\geq - O\left(\log N\right)\\
    \rho S(\Si) + O\left(\log N\right)  &\geq \log |\algo G|! - o(|\algo G|\log |\algo G|) \\
    S(\Si)\log N  & \geq \Omega (\log |\algo G|! - o(|\algo G|\log |\algo G|) )
    \end{align*}
    where the second and the last equality comes from the fact that $\epsilon = \omega(1/N)$ and $\rho = \Omega(\log N)$, respectively. 
    Since $\log |\algo G|! = O (|\algo G|\log|\algo G|)$, it follows that
    \begin{align*}
        S(\Si)\log N  & \geq \Omega (O(|\algo G| \log |\algo G|) - o(|\algo G|\log |\algo G|) )\\
        S(\Si)\log N  & \geq \Omega (|\algo G| \log |\algo G|).
    \end{align*}
    
    As we are conditioning on the event that $\algo G$ is large,
    plugging in the lower bound on $|\algo G|$, we have that, for sufficiently large $N$, $S(\Si) \geq \widetilde{\Omega}  (|\algo G|)$, and thus
    \begin{align*}
        S(\Si) \cdot T(\Si)^2 &\geq \widetilde{\Omega}(\epsilon N).
    \end{align*}
    This gives the desired space-time trade-off. 
      \end{proof}

We remark that the search inverter we consider in \Cref{thm:lower bounds for perm inversion} succeeds on more than just a constant number of inputs, that is $\epsilon = \omega(1/N)$, and beats the time complexity of $T = \Omega(\sqrt{\epsilon N})$ which is required for unstructured search using Grover's algorithm.~\cite{grover1996fast, dohotaru2008exact,zhandry2019record}. 
Next, we remove the restriction on the inverter by applying amplification (specifically, \expref{Corollary}{lem:amplification}.) This yields a lower bound in the full average-case version of the search inversion problem.
 
\begin{theorem}\label{thm:lower bounds for perm inversion2}
Let $\Si$ be a $(S, T, \epsilon)$-\spi for some $\epsilon > 0$.
Suppose that $\epsilon = \omega(1/N)$, $T= o(\epsilon^2 \sqrt{N})$, and $S\geq 1.$
Then, for sufficiently large $N$ we have
    $$
    S(\Si) \cdot T(\Si)^2 \geq \widetilde{\Omega}(\epsilon^3 N).
    $$
\end{theorem}

\begin{proof}
    Let $\Si = (\Si_0, \Si_1)$ be an $(S,T,\epsilon)$-\spi , for some $\epsilon > 0$. Using \expref{Corollary}{lem:amplification}, we can construct an \spi $\Si[\ell] = (\Si[\ell]_0, \Si[\ell]_1)$ with space and time complexities 
    $$
    S(\Si[\ell]) = \left\lceil\frac{\ln(10)}{\epsilon}\right\rceil \cdot S(\Si) \quad \text{ and } \quad T(\Si[\ell]) = \left(\left\lceil\frac{\ln(10)}{\epsilon}\right\rceil +1 \right) \cdot T(\Si)
    $$
    such that
    $$
    \Pr_{\pi,y} \left[\Pr_r\left[\pi^{-1}(y) \leftarrow\Si[\ell]_1^{\Opiii}(\Si[\ell]_0(\pi, r), y, r)\right] \geq \frac{2}{3} \right] \geq \frac{1}{5} .
    $$
    From \expref{Theorem}{thm:lower bounds for perm inversion} it follows that for sufficiently large $N \geq 1$,
    $$
    S(\Si[\ell]) \cdot T(\Si[\ell])^2 \geq \widetilde{\Omega}(N).
    $$
    Plugging in the expressions for $S(\Si[\ell])$ and $T(\Si[\ell])$, we get that with assumption $$\epsilon = \omega(1/N),\,\,\,\,\quad T(\Si)= o(\epsilon^2 \sqrt{N}) \quad \text{ and } \quad S(\Si)\geq 1,$$
    the trade-off between space and time complexities is $$S(\Si) \cdot T(\Si)^2 \geq \widetilde{\Omega}(\epsilon^3 N).$$
    \end{proof}

Note that we incur a loss ($\epsilon^3$ versus $\epsilon$) in our search lower bound due to the fact that we need to amplify the \emph{restricted} search inverter in \Cref{thm:lower bounds for perm inversion}. This results in a multiplicative overhead of $\Theta(1/\epsilon)$ in terms of space and time complexity, as compared to the restricted inverter.  We remark that a similar loss as a result of amplification is also inherent in~\cite{hhan2019quantum}. 

\subsection{Decision version}\label{decision-lower-bound}

\subsubsection{Space-time tradeoff, no adaptive sampling}

The search lower bound of \expref{Theorem}{thm:lower bounds for perm inversion2}, when combined with the search-to-decision reduction of \expref{Theorem}{thm: search-to-decision reduction}, yields a lower bound for the decision version.
\begin{corollary} \label{coro: decision permutation inverters with advice}
Let $\Di$ be a $(S, T, \delta)$-\dpi for some $\delta > 0$.
Suppose that $\delta = \omega(1/N)$ and $T= \Tilde{o}\left(\delta^2  \sqrt{N}\right)$  and $S\geq 1.$ Then, for sufficiently large $N$ we have
    $$
    S(\Di)\cdot T(\Di)^2 \gtrapprox \widetilde{\Omega}\left(\delta^6 N\right). 
    $$
\end{corollary}

\begin{proof}
  Let $N=2^n$. Given a $(S(\Di),T(\Di),\delta)$-\dpi $=(\Di_0, \Di_1)$ where $\Di_0$ outputs $S$-qubit state and $\Di_1$ makes $T$ queries, one can construct an $
  (S(\Si),T(\Si),\eta)$-\spi $=(\Si_0, \Si_1)$ by \expref{Theorem}{thm: search-to-decision reduction} with $\eta \geq 1-\textsf{negl}(n)$, and with space and time complexities $$
  S(\Si) = n \ell S(\Di) \quad \text{ and } \quad T(\Si) = n \ell T(\Di)
  $$
  where $\ell = \Omega\left(\frac{n(1+2\delta)}{\delta^2}\right)$.
  It directly follows from \expref{Theorem}{thm:lower bounds for perm inversion2} that with conditions 
  \begin{align*}
      \delta &= \omega(1/N), \qquad S(\Di)\geq 1,\\
      T(\Di)&= \frac{1}{n \ell} \cdot o(\eta \sqrt{N}) =  o\left(\frac{\delta^2}{n^2 (1+2\delta)}  \sqrt{N}\right) = \Tilde{o}\left(\delta^2 \sqrt{N}\right), 
  \end{align*}

  $\Si$ satisfies the space-time trade-off lower bound
      \begin{align*}
          n^3 \left(\frac{n(1+2\delta)}{\delta^2}\right)^3 S(\Di) \cdot  T(\Di)^2 &\geq \widetilde{\Omega}(\eta^3 N) \approx \widetilde{\Omega}(N)\\
          S(\Di) \cdot  T(\Di)^2 &\gtrapprox \widetilde{\Omega}\left(\delta^6 N\right)
      \end{align*}
      
  for sufficiently large $N$.
\end{proof}

Similar to the search lower bound from before, we incur a loss that amounts to a factor $\delta^6$. This results from our specific approach which is based on the search-to-decision reduction in \expref{Theorem}{thm: search-to-decision reduction}. We believe that our lower bound could potentially be improved even further.

\subsubsection{Time lower bound, adaptive sampling}

In the case of an adaptive decision inverter without advice, we can get a tight bound by means of the reduction from the unique search problem (\expref{Theorem}{thm:Nayak}), combined with well-known lower bounds on the average-case unique search problem.

\begin{theorem}\label{thm: decision permutation inverters without advice}
Let $\Di$ be a $(0, T, \delta)$-\adpi. Then $T^2 \geq \Omega(\delta N/M).$
\end{theorem}
\begin{proof}
  Since $\Di$ is a $(0, T, \delta)$-\dpi, 
  by the lower bound of unique search problem~\cite{grover1996fast,zalka1999grover, nayak2010inverting,zhandry2019record}, we get a $2T$-query algorithm for $\textsf{UNIQUESEARCH}_{n-1}$ with distributional error ($\frac{1}{2}-\delta$, $\frac{1}{2}$). 
  Since the YES and NO cases are uniformly distributed, we can write the overall error probability as 
  $\frac{1}{2}\left(\frac{1}{2}-\delta\right)+\frac{1}{2}\cdot\frac{1}{2}=\frac{1}{2}-\frac{\delta}{2}$.
  Then by 
  the lower bound of unique search, we have 
\begin{align*}
  1-\left(\frac{1}{2}-\frac{\delta}{2}\right) &\leq \frac{1}{2} + O\left(\frac{(2T)^2}{2^{n-m}}\right)\\
  T^2 &\geq \Omega(\delta \cdot 2^{n-m})\\
  T^2 &\geq \Omega\left(\frac{\delta N}{M}\right).
\end{align*}

We note that with non-adaptive $\Di$, i.e. $m=0$, the above bound reduces to query lower bound $T^2 \geq \Omega \left(\delta N\right)$.
  \end{proof}    

\section{Applications} \label{sec:applications}

In this section, we give a plausible security model for symmetric-key encryption and a scheme whose security in that model is based on the hardness of our adaptive two-sided permutation inversion problem. Recall that a symmetric-key encryption scheme consists of three algorithms:
\begin{itemize}
\item (key generation) $\Gen$: given randomness $s$ and security parameter $n$; outputs key $k:=\Gen(1^n; s)$;
\item (encryption) $\Enc$: given key $k$, plaintext $m$, and randomness $r$; outputs ciphertext $c := \Enc_k(m;r)$;
\item (decryption) $\Dec$: given key $k$, ciphertext $c$; outputs plaintext $m := \Dec_k(c)$.
\end{itemize}
When the key randomness is to be selected uniformly, we suppress it and simply write $\Gen(1^n)$.

Consider the following security definition.

\begin{definition}{($\sf OW$-$\sf QCCRA2$)}\label{def: OW-QCCRA2}
Let $\mathsf{SKE}=(\Gen, \Enc, \Dec)$ be a private-key encryption scheme. We say that $\mathsf{SKE}$ is $\sf OW$-$\sf QCCRA2$ if the advantage for any quantum polynomial-time adversary $\A$ in the following  \textsf{OW-QCCRA2} experiment  is at most negligible:
\begin{enumerate}
    \item A key $k$ is generated by running $\Gen(1^n;s)$;
    \item $\A$ gets quantum oracle access to $\Enc_k(\,\cdot\,;\,\cdot\,)$ and $\Dec_k(\cdot)$, and then outputs a $(m-1)$-bit string $\mu$ and a quantum state $\rho$ with size $S$. Let $t(n)$ be the number of quantum queries that $\A$ makes in this phase. 
    \item Uniform $b \in \{0,1\}$ and $r \in \{0,1\}^{n-1}$ are chosen, and a challenge ciphertext $c=\Enc_k(b\|\mu;r)$ is computed and given to $\A$;
    \item $\A$ gets quantum oracle access to $\Enc_k(\,\cdot\,;\,\cdot\,)$ and $\Dec_k^{\bot c}(\cdot)$, and eventually outputs a bit $b'$. Let $\ell(n)$ be the number of quantum queries that $\algo A$ makes in this phase. 
    \item The experiment outputs $1$, if $b'=b$, and $0$ otherwise.
\end{enumerate}   
\end{definition}
\ignore{
\begin{definition} \label{def: new-name}{(\textsf{new-name})} Let $\mathsf{SKE}=(\Gen, \Enc, \Dec)$ be a private-key encryption scheme. We say that $\mathsf{SKE}$ is \textsf{new-name}-secure if the advantage for any quantum polynomial-time adversary $\A$ in the following experiment is at most negligible:
\begin{enumerate}
    \item A key $k$ is generated by running $\Gen(1^n)$;
    \item $\A$ gets limited quantum oracle access to $\Enc_k(\,\cdot\,;\,\cdot\,)$ and $\Dec_k(\cdot)$; $\A$ then outputs $(m-1)$-bit string $\mu$. \ks{T(n) queries in the first phase} \ks{no memory being kept here}
    \item Uniform $b \in \{0,1\}$ and $r \in \{0,1\}^{n-1}$ are chosen, and a challenge ciphertext $c=\Enc_k(b\|\mu;r)$ is computed and given to $\A$;
    \item $\A$ gets quantum oracle access to $\Enc_k(\,\cdot\,;\,\cdot\,)$ and $\Dec_k^{\bot c}(\cdot)$. Eventually, it outputs a bit $b'$. \ks{L(n) queries in the first phase}
    \item The experiment outputs $1$, if $b'=b$, and $0$ otherwise.
\end{enumerate}
\end{definition}
}
We remark that, unlike in most definitions of security, here the adversary is allowed to choose both inputs to the encryption oracle: the plaintext as well as the randomness. To generate the challenge ciphertext, the coin $r$ needs to be chosen truly randomly; otherwise, the scheme will degenerate into a deterministic one that cannot be secure. Moreover, we do not yet make any restriction on the computational power of $\algo A$, or on the functions $t$ and $\ell$.

Next, we define two simple encryption schemes.

\medskip\noindent \textbf{RP Scheme.} Consider the following (inefficient) scheme that uses uniformly random permutations.
\begin{itemize}
    \item \Gen is given $1^n$ and outputs a description $k$ of a uniformly random permutation $\pi$ on $\bool^{2n}$;
    \item \Enc is given $k$, $m \in \bool^n$ and $r \in \bool^n$, and outputs $c:=\pi(m||r)$;
    \item \Dec is given $k$ and $c \in \bool^{2n}$, and outputs the first $n$ bits of $\pi^{-1}(c)$. 
\end{itemize}

\begin{definition}{($\epsilon$-Qsecure PRP)\cite{katz2020introduction,zhandry2016note}}\label{def: qPRP}
Let $P_k: \bool^\lambda \times \bool^n \rightarrow \bool^n$ be a permutation family. We call $P_k$ a $\epsilon$-Qsecure PRP if for any efficient quantum adversary $\algo A$ who makes $q$ quantum queries, there exist a negligible function $\epsilon(\lambda)$ such that
\begin{align*}
\left|\operatorname{Pr}\left[\algo A^{P_k(\cdot),P_k^{-1}(\cdot)}\left(1^n\right)=1\right]-\operatorname{Pr}\left[\algo A^{\pi(\cdot),\pi^{-1}(\cdot)}\left(1^n\right)=1\right]\right| \leq \epsilon\cdot \poly(q)\,,
\end{align*}
where $\pi: \bool^n \rightarrow \bool^n$ is a truly random permutation.
\end{definition}

\medskip\noindent \textbf{PRP Scheme.} Let $\{P_k: \bool^{2n} \mapsto \bool^{2n}\}$ be a family of  $\epsilon$-Qsecure PRPs and consider the following scheme:
\begin{itemize}
    \item \Gen takes as input a security parameter $1^n$ and returns a key $k \in \bool^n$ for $P_k$;
    \item \Enc is given key $k \in \bool^n$, $m\in \bool^n$ and $r \in \bool^n$, and outputs $c:=P_k(m||r)$;
    \item \Dec is given key $k \in \bool^n$ and $c \in \bool^{2n}$, and outputs the first $n$ bits of $P_k^{-1}(c)$. 
\end{itemize}

Of course, any practical scheme should be efficient, and indeed we can show that the PRP scheme is \textsf{OW-QCCRA2} in two special cases: when there is no advice, i.e., $S=0$ (we call this \textsf{OW-QCCRA2-v1})  and when there is no adaptivity, i.e., $|\mu|=0$ (we call this \textsf{OW-QCCRA2-v2}). We are able to prove the following theorems.

\begin{theorem} \label{thm: prp-newname-v1}
The \emph{PRP} scheme is $\sf OW$-$\sf QCCRA2$-$\sf v1$. In other words, for any quantum 
adversary $\mathcal{A}$ who makes $t(n)$ quantum queries in the pre-challenge phase and $\ell(n)$ quantum queries in the post-challenge phase, it holds that
    \begin{align*}
        \Pr[\mathsf{Exp}^{\mathsf{OW\mbox{-}QCCRA2\mbox{-}v1} }_{\A,\emph{PRP}}(1^n)=1] \leq \frac{1}{2}+\delta+ \epsilon \cdot T(n).
    \end{align*}
Here, $\delta \leq O(\frac{\ell^22^{n-1}}{2^{2n}})$, $T(n)=t(n)+\ell(n)$ and $\epsilon$ is a negligible function. 
\end{theorem}

\begin{proof} 
  Given an adversary $\algo A$ that attacks the RP scheme in the \textsf{OW-QCCRA2} experiment described in \expref{Definition}{def: OW-QCCRA2} with $S=0$, we can construct a $(0,T,\delta )$-\adpi $\aD = (\aD_0, \aD_1)$ in the decision inversion experiment, which takes place as follows:
\begin{enumerate}
\item \textbf{(sample instance and coins)} a random permutation $\pi: \bit^n \rightarrow \bit^n$ is sampled;
\item \textbf{(prepare advice)} $\aD_0$ is given the whole permutation table of $\pi$. Then it constructs oracles $\Enc(\cdot ; \cdot) = \pi(\cdot \| \cdot)$ and $\Dec(\cdot) = \pi^{-1}(\cdot)$ and gives $\algo A$ quantum oracle access. $\aD_0$ will get back a $(n-1)$-bit output string $\mu$ and then output it. Suppose $\algo A$ makes $t(n)$ quantum queries. 
\item \textbf{(invert)} An instance $c=\pi(b\|\mu\|r)$ is computed, with $b \in \{0,1\}$ and $r \in \{0,1\}^n$ are sampled. $\aD_1$ is run with $c$, auxiliary string $\mu$ and quantum oracle access  $\algo O_{\pi}$ and $\algo O_{\pi_{\bot y}^{-1}}$. It then directly passes $c$ and two oracles to $\algo A$ and gets back a bit $b'$ and outputs it. Suppose $\algo A$ makes $\ell(n)$ quantum queries. 
\item \textbf{(check)} If $b'=b$, output $1$; otherwise output $0$.
\end{enumerate}
It trivially follows that 
\begin{align*}
  \Pr[\textsf{Exp}^{\textsf{OW-QCCRA2-v1}}_{\A,\text{RP}}(1^n)=1] \leq \Pr[\DecisionInvert_{\aD}=1].
\end{align*}
By assumption we have that, for all efficient quantum adversary $\algo A$, there exists a negligible $\epsilon$ such that 
\begin{align*}
\left|\operatorname{Pr}\left[\algo A^{P_k(\cdot),P_k^{-1}(\cdot)}\left(1^n\right)=1\right]-\operatorname{Pr}\left[\algo A^{\pi(\cdot),\pi^{-1}(\cdot)}\left(1^n\right)=1\right]\right| \leq \epsilon \cdot \poly(t(n)+\ell(n)),
\end{align*}

Therefore 
\begin{align*}
  \Pr[\textsf{Exp}^{\textsf{OW-QCCRA2-v1}}_{\A,\text{PRP}}(1^n)=1] &\leq \Pr[\textsf{Exp}^{\textsf{OW-QCCRA2-v1}}_{\A,\text{RP}}(1^n)=1] +\epsilon \cdot T(n)\\
  &\leq \Pr[\DecisionInvert_{\aD}=1]+\epsilon \cdot T(n)\\
  &= \frac{1}{2} + \delta + \epsilon T(n).
\end{align*}
Where $\delta \leq O(\frac{\ell^22^{n-1}}{2^{2n}})$ by \expref{Theorem}{thm: decision permutation inverters without advice}, and by \expref{Definition}{def: qPRP} $\epsilon$ is negligible. Remark that the above bound becomes $\frac{1}{2}+\negl(n)$ when $\algo A$ is a quantum polynomial time (QPT) adversary since both $\delta$ and $\epsilon T$ are negligible when $t$ and $\ell$ are of polynomial size.

 \end{proof}
\begin{theorem} \label{thm: prp-newname-v2}
The \emph{PRP} scheme is $\sf OW$-$\sf QCCRA2$-$\sf v2$. In other words, for any quantum 
adversary $\mathcal{A}$ who makes $t(n)$ quantum queries in the pre-challenge phase and $\ell(n)$ quantum queries in the post-challenge phase, it holds that
    \begin{align*}
        \Pr[\mathsf{Exp}^{\mathsf{OW\mbox{-}QCCRA2\mbox{-}v1} }_{\A,\emph{PRP}}(1^n)=1] \leq \frac{1}{2}+\delta+ \epsilon \cdot T(n).
    \end{align*}
Here, $\delta \leq O(\frac{\ell^2 S}{2^{2n}})^{\frac{1}{6}}$, $T(n)=t(n)+\ell(n)$ and $\epsilon$ is a negligible function. 
\end{theorem}
\begin{proof}
Given an adversary $\algo A$ that attacks the RP scheme in the \textsf{OW-QCCRA2} experiment described in \expref{Definition}{def: OW-QCCRA2} with $|\mu|=0$, we can construct a $(S,T,\delta )$-\dpi $\sf D = (\sf D_0, \sf D_1)$ in the decision inversion experiment. The construction is the same as \expref{Theorem}{thm: prp-newname-v1}, with slight modifications at the "prepare advice" and the "invert" step:

\medskip\noindent \textbf{(prepare advice)} $\sf D_0$ is given the whole permutation table of $\pi$. Then it constructs oracles $\Enc(\cdot ; \cdot) = \pi(\cdot \| \cdot)$ and $\Dec(\cdot) = \pi^{-1}(\cdot)$ and gives $\algo A$ quantum oracle access. $\sf D_0$ will get back a $S$-qubit quantum state $\rho$ and then output it. Suppose $\algo A$ makes $t(n)$ quantum queries.

\medskip\noindent \textbf{(invert)} An instance $c=\pi(b||r)$ is computed, with $b \in \{0,1\}$ and $r \in \{0,1\}^n$ are sampled. $\sf D_1$ is run with $c$, quantum advice $\rho$ and quantum oracle access  $\algo O_{\pi}$ and $\algo O_{\pi_{\bot y}^{-1}}$. It then directly passes $c$ and two oracles to $\algo A$ and gets back a bit $b'$ and outputs it. Suppose $\algo A$ makes $\ell(n)$ quantum queries.

By following the same procedure as in \expref{Theorem}{thm: prp-newname-v1} but using the bound of \expref{Corollary}{coro: decision permutation inverters with advice}, we get the desired bound.    
\end{proof}

Finally, we remark that the above results hold for the following strengthening of \textsf{OW-QCCRA2}, described as follows. Suppose that an encryption scheme satisfies the property that there exists an \emph{alternative} decryption algorithm that can both compute the plaintext and also deduce the randomness that was initially used to encrypt. This property is true for the RP and PRP schemes, as well as some other standard encryption methods (e.g., Regev's secret-key \textsf{LWE} scheme, implicit in~\cite{regev2009lattices}). For schemes in this category, one can also grant access to such an alternative decryption algorithm, thus expanding the form of ``randomness access'' that the adversary has. Our proofs show that the RP and PRP schemes are secure (in their respective setting) even against this form of additional adversarial power.

\section{Future Work}

For future applications, the two-sided permutation inversion problem appears naturally in the context of sponge hashing~\cite{guido2011cryptographic} which is used by the international hash function
standard SHA3~\cite{dworkin2015sha}.
Previous work~\cite{czajkowski2018post,czajkowski2021quantum} studied the post-quantum security of the sponge construction where the block function is either a random function or a (non-invertible)
random permutation. However, as the core permutation in SHA3 is public and efficiently invertible, the ``right setting'' of theoretical study is one in which the block function consists of an invertible permutation. This setting is far less understood, and establishing the security of the sponge in this setting is a major open problem in post-quantum cryptography. Our results on two-sided permutation inversion may serve as a stepping stone towards this goal.

\printbibliography
\appendix


\section{Some basic probabilistic lemmas}

In this section we collect a series of known probabilistic results, which we used in our main proofs.

We first record some basic lemmas about the behavior of certain types of random variables.

\begin{lemma} [Multiplicative Chernoff Bound] \label{lma:chernoff}
Let $X_1,\dots,X_n$ be independent random variables taking values in $\bit$. Let $X = \sum_{i \in [n]} X_i$ denote their sum and let $\mu = \mathbb{E}[X]$ denote its expected value. Then for any $\delta > 0$,
$$
\Pr[ X < (1-\delta) \mu] \leq 2e^{-\delta^2\mu/2}.
$$
Specifically, when $X_i$ is a Bernoulli trial and $X$ follows the binomial distribution with $\mu = np$ and $p > \frac{1}{2}$, we have $\Pr[X \leq n/2] \leq e^{-n(p-\frac{1}{2})^2/(2p)}.$

\end{lemma}

\begin{lemma}[Reverse Markov's inequality]\label{lem:Reverse Markov's inequality}
Let $X$ be a random variable taking values in $[0,1]$. Let $\theta \in (0,1)$ be arbitrary. Then, it holds that
$$
\Pr[X \geq \theta] \geq \frac{\mathbb{E}[X] - \theta}{1-\theta}.
$$
\end{lemma}
\begin{proof}
Fix $\theta \in (0,1)$. We first show that
\begin{align}\label{eq:indicator}
(1-\theta) \cdot \id_{[X \geq \theta]}  \geq X - \theta,
\end{align}
where $\id_{[X \geq \theta]}$ is the indicator function for the event that $X \geq \theta$.
Suppose that $X \geq \theta$. Then, Eq.~\eqref{eq:indicator} amounts to $1-\theta \geq X - \theta$, which is satisfied because $X \leq 1$. Now suppose that $X < \theta$. In this case Eq.~\eqref{eq:indicator} amounts to $0 \geq X - \theta$, which is satisfied whenever $X \geq 0$. Taking the expectation over Eq.~\eqref{eq:indicator} and noting that $\mathbb{E}[\id_{[X \geq \theta]}] = \Pr[X \geq \theta]$, we get
\begin{align*}
 (1-\theta) \cdot \Pr[X \geq \theta]  \geq \mathbb{E}[X] - \theta. 
\end{align*}
This proves the claim.
  \end{proof}
\begin{lemma}[Averaging argument]\label{lem:averaging} Let $\algo X$ and $\algo Y$ be any finite sets and let $\Omega: \algo X \times \algo Y \rightarrow \bit$ be a predicate. Suppose that
$\Pr_{x,y}[\Omega(x,y)=1] \geq \epsilon$, for some $\epsilon \in [0,1]$, where $x$ is chosen uniformly at random in $\algo X$. Let $\theta \in (0,1)$. Then, there exists a subset $\algo X_\theta \subseteq \algo X$ of size $|\algo X_\theta| \geq (1-\theta) \cdot\epsilon |\algo X|$ such that
$$
\Pr_{y}[\Omega(x,y)=1] \geq \theta \cdot \epsilon, \quad \forall x \in \algo X_\theta.
$$
\end{lemma}
\begin{proof}
Define $p_x = \Pr_{y}[\Omega(x,y)=1]$, for $x \in \algo X$. Then, for $\epsilon \in [0,1]$, we have
$$
\mathbb{E}_x[p_x] = \Pr_{x,y}[\Omega(x,y)=1] = |\algo X|^{-1}\sum_{x \in \algo X} \Pr_{y}[\Omega(x,y) =1] \geq \epsilon. 
$$
Fix $\theta \in (0,1)$. Because the weighted average above is at least $\epsilon$, there must exist a subset $\algo X_\theta$ such that
$$
p_x = \Pr_{y}[\Omega(x,y)=1] \geq \theta \cdot \epsilon, \quad \forall x \in \algo X_\theta.
$$
Recall that $x$ is chosen uniformly at random in $\algo X$. Using the reverse Markov's inequality, 
it follows that
$$
\frac{|\algo X_\theta|}{|\algo X|} =
\Pr[p_x \geq \theta \cdot \epsilon] \geq \frac{\mathbb{E}[p_x] - \theta \cdot \epsilon}{1-\theta \cdot \epsilon} \geq \frac{\epsilon \cdot (1- \theta)}{1-\theta \cdot \epsilon} > \epsilon \cdot (1 - \theta).
$$
In other words, the subset $\algo X_\theta \subseteq \algo X$ is of size at least $|\algo X_\theta| \geq (1-\theta) \cdot\epsilon |\algo X|$.
  \end{proof}

\section{Amplification proofs} \label{app:amplification}

\subsection{Quantum oracle construction in \expref{Protocol}{ptc:eps-SPI repetition}} \label{app:search}
In \expref{Protocol}{ptc:eps-SPI repetition} step $2(c)$, $\Si[\ell]_1$, with quantum oracle access to $\Opi, \Opii$, needs to grant $\Si_1$ quantum oracle access to $(\sigma_{1,i}\circ \pi \circ \sigma_{2,i})_{\bot \sigma_{1,i}(y)}$, which is a simplified notation of $\algo O_{\sigma_{1,i}\circ \pi \circ \sigma_{2,i}}$ and $\algo O_{(\sigma_{1,i}\circ \pi \circ \sigma_{2,i})^{-1}_{\bot \sigma_{1,i}(y)}}$. Here we give detailed constructions of these two oracles:
\begin{itemize}
        \item Whenever the algorithm $\Si_1$ queries the oracle $\algo O_{\sigma_{1,i}\circ \pi \circ \sigma_{2,i}}$ on $\ket{w}_1\ket{z}_2$, $\Si[\ell]_1$ performs the following reversible operations
        \begin{align*}
            & \ket{w}_1\ket{z}_2\\
             \xrightarrow[]{\text{add \textsf{aux} registers}}&\ket{w}_1\ket{z}_2\ket{0}_{\mathsf{aux1}}\ket{0}_{\mathsf{aux2}}\\
             \xrightarrow[]{\algo O_{\sigma_{2,i},1,\mathsf{aux2}}} &\ket{w}_1\ket{z}_2\ket{0}_{\mathsf{aux1}}\ket{\sigma_{2,i}(w)}_{\mathsf{aux2}}\\
             \xrightarrow[]{\algo O_{\pi,\mathsf{aux2},\mathsf{aux1}}} &\ket{w}_1\ket{z}_2\ket{\pi \circ \sigma_{2,i}(w)}_{\mathsf{aux1}}\ket{\sigma_{2,i}(w)}_{\mathsf{aux2}}\\
             \xrightarrow[]{\algo O_{\sigma_{1,i}, \mathsf{aux1},2}} &\ket{w}_1\ket{z\oplus \sigma_{1,i}\circ \pi \circ \sigma_{2,i}(w)}_2\ket{\pi \circ \sigma_{2,i}(w)}_{\mathsf{aux1}}\ket{\sigma_{2,i}(w)}_{\mathsf{aux2}}\\
             \xrightarrow[]{\algo O_{\pi,\mathsf{aux2},\mathsf{aux1}}} &\ket{w}_1\ket{z\oplus \sigma_{1,i}\circ \pi \circ \sigma_{2,i}(w)}_2\ket{0}_{\mathsf{aux1}}\ket{\sigma_{2,i}(w)}_{\mathsf{aux2}}\\
             \xrightarrow[]{\algo O_{\sigma_{2,i},1,\mathsf{aux2}}} &\ket{w}_1\ket{z\oplus \sigma_{1,i}\circ \pi \circ \sigma_{2,i}(w)}_2\ket{0}_{\mathsf{aux1}}\ket{0}_{\mathsf{aux2}}\\
             \xrightarrow[]{\text{drop } \mathsf{aux}} &\ket{w}_1\ket{z\oplus \sigma_{1,i}\circ \pi \circ \sigma_{2,i}(w)}_2.
        \end{align*}
        Then, $\Si[\ell]_1$  sends  the final state back to $\Si_1$.
        
        \item Whenever $\Si_1$ queries the oracle $\algo O_{(\sigma_{1,i}\circ \pi \circ \sigma_{2,i})^{-1}_{\bot \sigma_{1,i}(y)}}$ on $\ket{w}_1\ket{z}_2$, the algorithm $\Si[\ell]_1$ performs the following reversible operations:
        \begin{align*}
            & \ket{w}_1\ket{z}_2\\
             \xrightarrow[]{\text{add \textsf{aux} register}}&\ket{w}_1\ket{z}_2\ket{0}_{\mathsf{aux1}}\ket{0}_{\mathsf{aux2}}\\
             \xrightarrow[]{\algo O_{\sigma_{1,i,*}^{-1},1,\mathsf{aux1}}} &\ket{w}_1\ket{z}_2\ket{\sigma_{1,i,*}^{-1}(w)}_{\mathsf{aux1}}\ket{0}_{\mathsf{aux2}}\\
             \xrightarrow[]{\algo O_{\pi_{\bot y}^{-1},\mathsf{aux1},\mathsf{aux2}}} &\ket{w}_1\ket{z}_2\ket{\sigma_{1,i,*}^{-1}(w)}_{\mathsf{aux1}}\ket{\pi^{-1}_{\bot y} \circ \sigma_{1,i,*}^{-1}(w)}_{\mathsf{aux2}}\\
             \xrightarrow[]{\algo O_{\sigma_{2,i,*}^{-1},2,\mathsf{aux2}}} &\ket{w}_1\ket{z\oplus \sigma_{2,i,*}^{-1} \circ \pi^{-1}_{\bot y} \circ \sigma_{1,i,*}^{-1}(w) }_2\ket{\sigma_{1,i,*}^{-1}(w)}_{\mathsf{aux1}}\ket{\pi^{-1}_{\bot y} \circ \sigma_{1,i,*}^{-1}(w)}_{\mathsf{aux2}}\\
             \xrightarrow[]{\algo O_{\pi_{\bot y}^{-1},\mathsf{aux1},\mathsf{aux2}}} &\ket{w}_1\ket{z\oplus \sigma_{2,i,*}^{-1} \circ \pi^{-1}_{\bot y} \circ \sigma_{1,i,*}^{-1}(w) }_2\ket{\sigma_{1,i,*}^{-1}(w)}_{\mathsf{aux1}}\ket{0}_{\mathsf{aux2}}\\
             \xrightarrow[]{\algo O_{\sigma_{1,i,*}^{-1},1,\mathsf{aux1}}} &\ket{w}_1\ket{z\oplus \sigma_{2,i,*}^{-1} \circ \pi^{-1}_{\bot y} \circ \sigma_{1,i,*}^{-1}(w) }_2\ket{0}_{\mathsf{aux1}}\ket{0}_{\mathsf{aux2}}\\
             \xrightarrow[]{\text{drop } \mathsf{aux}} &\ket{w}_1\ket{z\oplus \sigma_{2,i,*}^{-1} \circ \pi^{-1}_{\bot y} \circ \sigma_{1,i,*}^{-1}(w) }_2.
        \end{align*}
        where $\sigma^{-1}_{\cdot,i,*}: [N] \times \bit \rightarrow [N] \times \bit$ is given below 
        \begin{align*}
            \sigma^{-1}_{\cdot,i,*} (w \| b) := \sigma_{\cdot,i}^{-1}(w) \| b.
        \end{align*}
        Then, $\Si[\ell]_1$ sends the final state back to $\Si_1$. 
        \end{itemize}

\subsection{Another amplification lemma proof} \label{app:amp3}

\begin{replemma}{lem:amplification}
Let $\Si = (\Si_0, \Si_1)$ be an $\epsilon$-\spi with space and time complexity given by $S(\Si)$ and $T(\Si)$, respectively, for some $\epsilon >0$. Then, we can construct an \spi $\Si[\ell] = (\Si[\ell]_0, \Si[\ell]_1)$ with space and time complexities
$$
S(\Si[\ell]) = \left\lceil\frac{\ln(10)}{\epsilon}\right\rceil \cdot S(\Si) \quad \text{ and } \quad T(\Si[\ell]) = \left\lceil\frac{\ln(10)}{\epsilon}\right\rceil \cdot (T(\Si)+1)
$$
such that
$$
\Pr_{\pi,y} \left[\Pr_r\left[\pi^{-1}(y) \leftarrow \Si[\ell]_1^{\Opiii}(\rho, y, r): \rho \leftarrow \Si[\ell]_0(\pi, r) \right] \geq \frac{2}{3} \right] \geq \frac{1}{5}.
$$
    
\end{replemma}

\begin{proof}
Let $\ell = \left\lceil\frac{\ln(10)}{\epsilon}\right\rceil$. Using \expref{Lemma}{lemma:amplify-S}, we can construct an $\ell$-time repetition of $\Si$ $(\eta)$-$\spi$, denoted by $\Si[\ell] = (\Si[\ell]_0,\Si[\ell]_1)$, with $\eta =  1-(1-\epsilon)^{\ell}$ and space and time complexities
$
S(\Si[\ell]) = \ell\cdot S(\Si)$ and $T(\Si[\ell]) =  \ell \cdot (T(\Si)+1)$.
In other words,
$$
\Pr_{\pi,y,r}\left[\pi^{-1}(y) \leftarrow \Si[\ell]_1^{\Opiii}(\rho, y, r): \rho \leftarrow \Si[\ell]_0(\pi, r) \right] \geq  1-(1-\epsilon)^{\ell} \geq \frac{9}{10}.
$$
Let $\algo S_N$ denote the set of permutations over $[N]$. From \expref{Lemma}{lem:averaging} it follows that there exists $\theta = 7/9$ and a subset $\algo X_\theta \subseteq \algo S_N \times [N]$ of size at least
$$
|\algo X_\theta| \geq (1-\theta)\cdot \frac{9}{10} \cdot \big|\algo S_N \times [N]\big| = \frac{1}{5} \cdot \big|\algo S_N \times [N] \big|.
$$
such that, for every $(\pi,y) \in \algo X_\theta$, we have
$$
\Pr_r\left[\pi^{-1}(y) \leftarrow\Si[\ell]_1^{\Opiii}(\rho, y, r): \rho \leftarrow \Si[\ell]_0(\pi, r)\right] \geq \theta \cdot \frac{9}{10} > \frac{2}{3}.
$$
Because $|\algo X_\theta| \cdot |\algo S_N \times [N]|^{-1} \geq  \frac{1}{5}$, it follows that
$$
\Pr_{\pi,y} \left[\Pr_r\left[\pi^{-1}(y) \leftarrow \Si[\ell]_1^{\Opiii}(\rho, y, r): \rho \leftarrow \Si[\ell]_0(\pi, r)\right] \geq \frac{2}{3} \right] \geq \frac{1}{5}.
$$
This proves the claim.
\end{proof}

\subsection{Decision amplification proof}\label{app:decision}

Same as the search amplification, we amplify the success probability of a $\delta$-\dpi through $\ell$-time repetition defined in \expref{Protocol}{ptc:eps-DPI repetition}.
\begin{protocol}[$\ell$-time repetition of \textsf{$\delta$-\dpi}] \label{ptc:eps-DPI repetition}
Given a \textsf{$\delta$-\dpi} $\Di = (\Di_0, \Di_1)$, the construction of an "$\ell$-time serial repetition of $\Di$" $\Di[\ell] = (\Di[\ell]_0, \Di[\ell]_1)$ is as follows:

\begin{enumerate}
   \item \textit{(Advice Preparation)} the algorithm $\Di[\ell]_0$ proceeds as follows:
    \begin{enumerate}
        \item $\Di[\ell]_0$ receives as input a random permutation $\pi: [N]\rightarrow [N]$ and randomness $r \leftarrow \bit^{*}$ and parses the string $r$ into $2\ell$ substrings, i.e. $r=r_{0}\Vert...\Vert r_{\ell-1}\Vert r_{\ell}\Vert...\Vert r_{2\ell-1}$ (the length is clear in context).
        \item $\Di[\ell]_0$ uses $r_{0},...,r_{\ell-1}$ to generate $\ell$ permutation pairs $\{\sigma_{1,i}, \sigma_{2,i}\}_{i=0}^{\ell-1}$ in $\Perms_N$, where $\sigma_{1,i}$ is a random permutation, $\sigma_{2,i}$ has the following form
        \begin{equation}\label{eqn:decision case}
        \sigma_{2,i} (x_1,...,x_n)= (x_1 \oplus r_i^*,x_2,...,x_n),
        \end{equation} where $r_i^*$ is some random bit generated from $r_i$ for all $i \in [0, \ell-1]$.
        Then runs $\Di_0(\sigma_{1,i}\circ \pi \circ \sigma_{2,i} , r_{i+\ell})$ to get a quantum state $\rho_i := \rho_{\sigma_{1,i}\circ \pi \circ \sigma_{2,i}, r_{i+\ell}}$ for all $i \in [0,\ell-1]$. Finally, $\Di[\ell]_0$ outputs a quantum state $\bigotimes_{i=0}^{\ell-1}\rho_i$.
    \end{enumerate}
    \item \textit{(Oracle Algorithm)} $\Di[\ell]_1^{\Opiii}$ is an oracle algorithm that proceeds as follows:
    \begin{enumerate}
    \item $\Di[\ell]_1$ receives $\bigotimes_{i=0}^{\ell-1}\rho_i$, randomness $r $ and an image $y \in [N]$ as input. 
    \item $\Di[\ell]_1$ parses $r = r_{0}\Vert...\Vert r_{\ell-1}\Vert r_{\ell}\Vert...\Vert r_{2\ell-1}$ and uses the coins $r_{0}\Vert...\Vert r_{\ell-1}$ to generate $\ell$ different permutation pairs $\{\sigma_{1,i}, \sigma_{2,i}\}_{i=0}^{\ell-1}$ in $\Perms_N$ as shown above.
    \item  $\Di[\ell]_1$ then runs the following routine for all $i\in[0,\ell-1]$:
    \begin{enumerate}
    \item Run $\Di_1$ with oracle access to $(\sigma_{1,i}\circ \pi \circ \sigma_{2,i})_{\bot \sigma_{1,i}(y)}$, which implements the permutation $\sigma_{1,i}\circ \pi \circ \sigma_{2,i}$ and its inverse (but $\bot$ at $\sigma_{1,i}(y)$).
      \item Get back $b_i \leftarrow \Di_1^{(\sigma_{1,i}\circ \pi \circ \sigma_{2,i})_{\bot \sigma_{1,i}(y)}}(\rho_i,\sigma_{1,i}(y),r_{i+\ell})$. 
    \end{enumerate}
      \item  $\Di[\ell]_1$ pads $b_i$ with all zero string of size $n-1$ and computes $b_i^*=\sigma_{2,i}(b_i\|0^{n-1})|_0$ for all $i \in [0,\ell-1]$, then outputs $b^*$ which is the majority vote of $\{b_0^*,\ldots,b_{\ell-1}^*\}$.
    \end{enumerate}
\end{enumerate}
\end{protocol}

\begin{replemma}{lemma:amplify-D}
 Let $(\Di_0, \Di_1)$ be a $\delta$-\dpi, where $\Di_0$ outputs an $S$-qubit state and $\Di_1$ makes $T$ queries. Then, we can construct an $\ell$-time repetition of $\Di$, denoted by $\Di[\ell] = (\Di[\ell]_0,\Di[\ell]_1)$, which is an \textsf{$\eta$-\dpi} for $\eta \geq \frac{1}{2}- \exp(-\frac{\delta^2}{(1+2\delta)} \cdot \ell)$, and has space and time complexities given by
$$
S(\Di[\ell]) = \ell \cdot S(\Di) \quad \text{ and } \quad T(\Di[\ell]) = \ell \cdot T(\Di).
$$
\end{replemma}

\begin{proof}
Let $(\Di_0, \Di_1)$ be a $\delta$-\dpi for some $\delta >0$, where $\Di_0$ outputs an $S$-qubit state and $\Di_1$ makes $T$ queries.
Similarly as in \expref{Lemma}{lemma:amplify-S}, we consider the execution of the $\ell$-time repetition of \textsf{$\delta$-DPI}, denoted by \textsf{DPI} $\Di[\ell]$, which we define in \expref{Protocol}{ptc:eps-DPI repetition}. For each iteration $i \in [0,\ell-1]$, we have
\begin{align*}
&\Pr[b_i = \pi^{-1}(y)|_0] \\
&=\Pr \left[ (\bar{\pi})^{-1}(\sigma_{1,i}(y)) |_0 \leftarrow
\Di_1^{(\bar{\pi})_{\bot \sigma_{1,i}(y)}}\big(\rho_i,\sigma_{1,i}(y),r_{i+\ell} \big) : \rho_i \leftarrow \Di_0(\bar{\pi}, r_{i+\ell})
\right] \\
&\equiv \Pr \left[ ((\sigma_{2,i})^{-1} \circ \pi^{-1}(y))|_0 \leftarrow
\Di_1^{\Opiii} \big(\rho_{\pi \circ \sigma_{2,i}, r_{i+ \ell}},y,r_{i+\ell} \big) :\rho_{\pi \circ \sigma_{2,i}, r_{i+ \ell}} \leftarrow \Di_0(\pi \circ \sigma_{2,i}, r_{i+\ell})
\right] \,\, \\
&\geq \,\, \frac{1}{2} + \delta, 
\end{align*}
where $\bar{\pi}= \sigma_{1,i}\circ\pi \circ \sigma_{2,i}$. The probability is taken over $\pi \leftarrow \Perms_N$, $r \leftarrow \bit^*$ (which is used to sample permutations $\sigma_i$) and $x \leftarrow [N]$, along with all internal measurements of $\Di$. 

Recall that $b_i \leftarrow \Di_1^{(\sigma_{1,i}\circ \pi \circ \sigma_{2,i})_{\bot \sigma_{1,i}(y)}}(\rho_i,\sigma_{1,i}(y),r_{i+\ell})$, for $i \in [\ell]$. Let $X_i$ be the indicator variable for the event that $b_i = (\pi\circ \sigma_{2,i})^{-1}(y)|_0$. Similar to the search case, we argue that all $X_i$ are mutually independent. For any $i \in [0, \ell-1]$ and any subset $K \subset [0,l-1]$ where $i \notin K$, let

\begin{align*}
    \textsf{Event } A &= \{X_i=0\} \\
      &= \{b_i \neq \big((\sigma_{2,i})^{-1} \circ \pi^{-1}(y)\big)|_0\} \\
      &= \{b_i||0^{n-1}|_0 \neq \big((\sigma_{2,i})^{-1} \circ \pi^{-1}(y)\big)|_0 \} \\
      &= \{\big(\sigma_{2,i} \circ (b_i||0^{n-1})\big)|_0 \neq \pi^{-1}(y)|_0 \},
\end{align*}

Note that the last equality holds because of \expref{Equation}{eqn:decision case}. We then define another event
\begin{align*}
    \textsf{Event } B &= \bigcap_{j\in K}\{ X_j=0 \} \\
    &= \bigcap_{j\in K}\{\big(\sigma_{2,j}\circ (b_j||0^{n-1})\big)|_0 \neq \pi^{-1}(y)|_0 \}
 \end{align*}

Given that $B$ happens, we have $\{b_j\}_{j \in K}$ such that for all $j\in K$, $\big(\sigma_{2,j}\circ (b_i||0^{n-1})\big)|_0 \neq \pi^{-1}(y)|_0$. We now consider the probability that $A$ happens. In \expref{Equation}{eqn:decision case}, since all $r_i^*$ are independently randomly generated, the value of $\big(\sigma_{2,i} \circ (b_i||0^{n-1})\big)|_0$ is independent of all other values of  $\big(\sigma_{2,j}\circ (b_i||0^{n-1})\big)|_0$. Therefore, the event that $\big(\sigma_{2,i} \circ (b_i||0^{n-1})\big)|_0 \neq \pi^{-1}(y)|_0$ is not correlated with all other $\big(\sigma_{2,j}\circ (b_i||0^{n-1})\big)|_0 \neq \pi^{-1}(y)|_0$, i.e., $\Pr[A|B]=\Pr[A]$. This is true for any $i$ and $K$. Same as the search case, in each trial, the base inverter is solving a completely independent permutation inversion problem, thus we conclude that all $\ell$ trails are mutually independent. 

Let $X=\sum_{i=0}^{\ell-1} X_i$, we have that $\mathbb{E}[X] \geq \ell \cdot (\frac{1}{2} + \delta)$ by the linearity of expectation. Note that $\Di[\ell]$ succeeds in $\DecisionInvert$ if and only if $\Di[\ell]_1$ can output $b^* = \pi^{-1}(y)|_0$, i.e. $X > \frac{\ell}{2}$ in which case more than half of the elements in $\{b_0,...,b_{\ell-1}\}$ are equal to $\pi^{-1}(y)|_0$. By the multiplicative Chernoff bound in \Cref{lma:chernoff}, the probability that $\DecisionInvert$ fails is at most
$$
 \mathrm{Pr}\left[X < \frac{\ell}{2}\right] \leq \exp\left(-\frac{\delta^2}{(1+2\delta)} \cdot \ell \right).
$$

Note that the resource requirements needed for the amplification procedure amount to space and time complexities $\ell S$ and $\ell T$, respectively, similar as in \Cref{lemma:amplify-S}.
\end{proof}

\section{Quantum oracle constructions in \expref{Theorem}{thm:Nayak}}\label{app:reduction_Nayak}

In \expref{Theorem}{thm:Nayak}, $\B$, with quantum oracle access to $f$, needs to grant $\A$ quantum oracle access to $h_{f,\pi,t,\mu}$ and $h^{-1*}_{f,\pi,t,\mu}$. Here we give detailed constructions of $\algo O_{h_{f,\pi,t,\mu}}$ and $\algo O_{h^{-1*}_{f,\pi,t,\mu}}$. Note that $\pi$ is sampled by $\B$ and so it is easy for it to construct quantum oracles $\Opi$ and $\algo O_{\pi^{-1}_{\bot t}}$. Since $h_{f,\pi,t,\mu}^{-1*} =\pi^{-1}_{\bot t}$, the partial inverse oracle $O_{h^{-1*}_{f,\pi,t,\mu}}$ can be simply simulated by $\algo O_{\pi^{-1}_{\bot t}}$. So we only need to show how to construct $\algo O_{h_{f,\pi,t,\mu}}$. 

Let $x=x_0\dots x_{n-1}$, where $n = \log N$. When $\pi \in \pi_{t,0,\mu}$, the function becomes
\begin{align*}
    h_{f, \pi,t,\mu}(x_0\dots x_{n-1})&=\left(x_0 \cdot f(x_1...x_{n-m-1})\cdot \mathds{1}(x_{n-m}...x_{n}=\mu) \right)\cdot t \\
    &+ \overline{ \left( x_0 \cdot f(x_1...x_{n-m-1})\cdot \mathds{1}(x_{n-m}...x_{n}=\mu)  \right)}\cdot \pi(x).
\end{align*}
Then define a function $g: [N]\rightarrow \{0,1\}$, such that $g(x)=x_0 \cdot f(x_1...x_{n-m-1})\cdot \mathds{1}(x_{n-m}...x_{n}=\mu) $. With access to $\algo O_f$, it is easy to construct $\algo O_g$ by applying $\algo O_f$ to the last $n-1$ bits followed by an AND gate.

Now when $\A$ queries the oracle $\algo O_{h_{f,\pi,t,\mu}}$ on $\ket{x}\ket{y}$, $\B$ performs the following reversible operations
\begin{align*}
& \ket{x}\ket{y}\\
\xrightarrow[]{\text{add \textsf{aux} registers}}&\ket{x}_1\ket{y}_2\ket{0}_3\ket{0}_4\ket{0^n}_5\ket{0^n}_6\\
\xrightarrow[]{\algo O_{g,1,3} X_4 \algo O_{1,4} \algo O_{\pi,1,5}U_t}&\ket{x}\ket{y}\ket{g(x)}\ket{\overline{g(x)}}\ket{\pi(x)}\ket{t}\\
\xrightarrow[]{\text{CCNOT}_{3,6,2}} &\ket{x}\ket{y\oplus (g(x)\cdot t)}\ket{g(x)}\ket{\overline{g(x)}}\ket{\pi(x)}\ket{t}\\
\xrightarrow[]{\text{CCNOT}_{4,5,2}} &\ket{x}\ket{y\oplus (g(x)\cdot t)\oplus (\overline{g(x)}\cdot \pi(x))}\ket{g(x)}\ket{\overline{g(x)}}\ket{\pi(x)}\ket{t}\\
\xrightarrow[]{\algo O_{g,1,3} X_4 \algo O_{1,4} \algo O_{\pi,1,5}U_t} &\ket{x}\ket{y\oplus (g(x)\cdot t)\oplus (\overline{g(x)}\cdot \pi(x))}\ket{0}\ket{0}\ket{0^n}\ket{0^n}\\
\xrightarrow[]{\text{drop \textsf{aux}}}&\ket{x}\ket{y\oplus (g(x)\cdot t)\oplus (\overline{g(x)}\cdot \pi(x))}
\end{align*}
It is easy to see that $y\oplus (g(x)\cdot t)\oplus (\overline{g(x)}\cdot \pi(x))=y\oplus h_{f,\pi,t,\mu}(x)$. Therefore, to respond to one query to $O_{h_{f,\pi,t,\mu}}$, $\B$ needs to query $\algo O_f$ \textit{twice} (once for computing and once for eliminating). The same thing can be done when $\pi\in \pi_{t,1,\mu}$. 

\section{Quantum oracle constructions in \expref{Protocol}{ptc:QRAC-perm}}\label{sec:oracle-for-pi-bar}

Here, we show how to implement the function $\bar{\pi}_{\bot y}$ by means of a (reversible) quantum oracle. This can be done by two separate oracles $\algo O_{\bar{\pi}}$ and $\algo O_{\bar{\pi}^{-1}_{\bot y}}$, where the corresponding functions are 
$$
\bar{\pi}(w)= \begin{cases}
y & \text{ if } w \in \algo G \\
\pi(w) & \text{ if } w \notin \algo G
\end{cases}
$$
and
$$
\bar{\pi}^{-1}_{\bot y}(w,b) = \begin{cases}
\pi^{-1}(w)||0 & \text{ if } w \notin \pi(\algo G)  \, \land \, b=0\\
1||1 & \text{ if } w \in \pi(\algo G) \, \land \, b=1.
\end{cases}
$$

Let $f$ be an indicator function on whether $w \in \mathcal{G}$. Given $\beta$ as an input, the permutation $\pi$ restricted to inputs outside of $\mathcal{G}$ is known (denoted as $\pi'$). Therefore given input $y$, with quantum oracle access to $\algo O_f$ and $\algo O_{\pi'}$, we can easily construct $\algo O_{\bar{\pi}}$ and $\algo O_{\bar{\pi}^{-1}_{\bot y}}$. 

The following procedure gives a construction of $\algo O_{\bar{\pi}}$.
\begin{align*}
    & \ket{w}\ket{z}\\
    \xrightarrow[]{\text{add \textsf{aux} registers}}&\ket{w}_1\ket{z}_2\ket{0}_3\ket{0}_4\ket{0^n}_5\ket{0^n}_6\\
    \xrightarrow[]{\algo O_{f,1,3} X_4 \algo O_{1,4} \algo O_{\pi',1,5}U_y}&\ket{w}\ket{z}\ket{f(w)}\ket{\overline{f(w)}}\ket{\pi'(w)}\ket{y}\\
    \xrightarrow[]{\text{CCNOT}_{3,6,2}} &\ket{w}\ket{z\oplus (f(w)\cdot y)}\ket{f(x)}\ket{\overline{f(w)}}\ket{\pi'(w)}\ket{t}\\
    \xrightarrow[]{\text{CCNOT}_{4,5,2}} &\ket{x}\ket{z\oplus (f(w)\cdot y)\oplus (\overline{f(w)}\cdot \pi'(w))}\ket{f(w)}\ket{\overline{f(w)}}\ket{\pi'(w)}\ket{y}\\
    \xrightarrow[]{\algo O_{f,1,3} X_4 \algo O_{1,4} \algo O_{\pi',1,5}U_y} &\ket{x}\ket{z\oplus (f(w)\cdot y)\oplus (\overline{f(w)}\cdot \pi'(w))}\ket{0}\ket{0}\ket{0^n}\ket{0^n}\\
    \xrightarrow[]{\text{drop \textsf{aux}}}&\ket{x}\ket{z\oplus (f(w)\cdot y)\oplus (\overline{f(w)}\cdot \pi'(w))}\\
    \equiv&\ket{x}\ket{z\oplus \bar{\pi}(w)}
\end{align*}

The backward oracle $\algo O_{\bar{\pi}^{-1}_{\bot y}}$ would be constructed similarly.

\end{document}